\documentclass[a4paper,12pt]{article}

\usepackage{amsmath,amssymb,amsthm}
\usepackage{authblk}
\usepackage[a4paper, margin=2.5cm]{geometry}
\usepackage{hyperref}
\hypersetup{
  pdftitle = {Enumerating Minimal Dominating Sets in Kt-free graphs and variants},
  pdfauthor=  {Marthe Bonamy, Oscar Defrain, Marc Heinrich, Michal Pilipczuk, Jean-Florent Raymond},
  colorlinks = true,
  citecolor = [rgb]{0,0.8,0}
}
\usepackage[utf8]{inputenc}
\usepackage[T1]{fontenc}
\usepackage{enumerate}
\usepackage[absolute]{textpos} 
\usepackage[noadjust]{cite}
\usepackage{needspace}
\usepackage{tikz}
\usetikzlibrary{decorations.pathreplacing}
\tikzstyle{small node} = [draw, circle, fill = black, minimum size = 3pt, inner sep = 0pt]
\tikzstyle{black node} = [draw, circle, fill = black, minimum size = 5pt, inner sep = 0pt]
\tikzstyle{white node} = [draw, circle, fill = white, minimum size = 5pt, inner sep = 0pt]
\tikzstyle{normal} = [draw=none, fill = none, rectangle, minimum size =0]

\usepackage{thmtools}
\usepackage{thm-restate}
\declaretheorem[name=Theorem, numberwithin=section]{theorem}
\declaretheorem[name=Lemma, sibling=theorem]{lemma}
\declaretheorem[name=Proposition, sibling=theorem]{proposition}
\declaretheorem[name=Observation, sibling=theorem]{observation}

\declaretheorem[name=Definition, sibling=theorem]{definition}
\declaretheorem[name=Corollary, sibling=theorem]{corollary}

\declaretheorem[name=Claim, sibling=theorem]{claim}
\declaretheorem[name=Remark, style=remark, sibling=theorem]{remark}

\def\cqedsymbol{\ifmmode$\lrcorner$\else{\unskip\nobreak\hfil
\penalty50\hskip1em\null\nobreak\hfil$\lrcorner$
\parfillskip=0pt\finalhyphendemerits=0\endgraf}\fi} 
\newcommand{\cqed}{\renewcommand{\qed}{\cqedsymbol}}

\def\ie{i.e.}
\def\C{\mathcal{C}} 
\def\Cp{\mathcal{C'}} 
\def\D{\mathcal{D}} 
\def\F{\mathcal{F}} 
\def\Dcs{\textsc{Dcs}} 
\def\SAT{\textsc{SAT}}
\def\DomEnum{\textsc{Dom-Enum}}
\newcommand{\transv}{\textsc{Trans-Enum}}
\newcommand{\N}{\mathbb{N}}
\newcommand{\algoa}{\texttt{A}} 
\newcommand{\algob}{\texttt{B}} 
\newcommand{\algobp}{\texttt{B'}} 
\renewcommand{\geq}{\geqslant}
\renewcommand{\leq}{\leqslant}

\newcommand{\setin}{\Sigma_{\rm in}}
\newcommand{\setout}{\Sigma_{\rm out}}
\newcommand{\intv}[2]{\left \{ #1, \dots, #2 \right \}}
\DeclareMathOperator{\parent}{\sf Parent}
\DeclareMathOperator{\priv}{Priv}
\DeclareMathOperator{\poly}{poly}

\def\TFNPFP{{\sf TFNP$=$FP}}
\def\NP{{\sf NP}}
\newcommand\eqdef{\overset{\text{\tiny{def}}}{=}} 

\title{Enumerating minimal dominating sets\\ in $K_t$-free graphs and variants\thanks{A preliminary version of this article appeared in the proceedings of the 36\textsuperscript{th} Symposium on Theoretical Aspects of Computer Science (STACS 2019)~\cite{bonamy2019triangle}. The first author has been supported by the ANR project GrR ANR-18-CE40-0032. The second author has been supported by the ANR project GraphEn ANR-15-CE40-0009.
The fourth author is supported by project TOTAL, which has received funding from the European Research Council (ERC) 
under the European Union's Horizon 2020 research and innovation programme, grant agreement No.~677651.
The last author has been supported by the ERC consolidator grant DISTRUCT-648527.}}

\author[1]{Marthe Bonamy}
\author[2]{Oscar Defrain}
\author[3]{Marc Heinrich}
\author[4]{\\Micha\l{} Pilipczuk}
\author[5]{Jean-Florent Raymond}
\affil[1]{CNRS, LaBRI, Université de Bordeaux, France.}
\affil[2]{LIMOS, Université Clermont Auvergne, France.}
\affil[3]{LIRIS, Université Claude-Bernard, Lyon, France.}
\affil[4]{University of Warsaw, Poland.}
\affil[5]{CNRS, LIMOS, Université Clermont Auvergne, France.}

\date{\today}

\begin{document}

\maketitle

\begin{textblock}{20}(0, 13.1)
\includegraphics[width=40px]{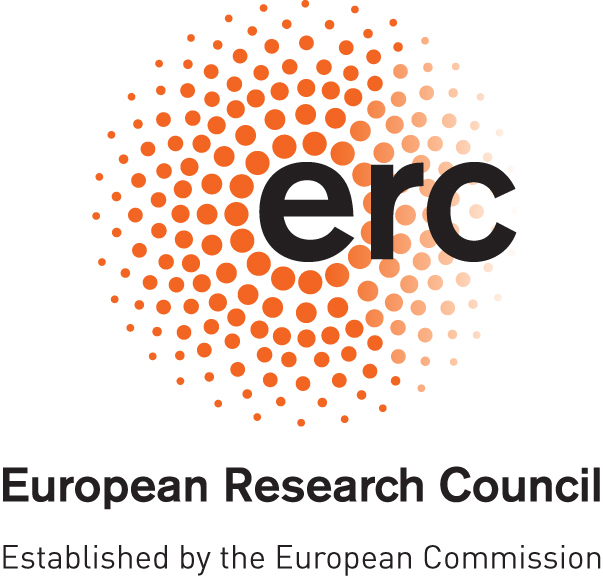}%
\end{textblock}
\begin{textblock}{20}(-0.35, 13.4)
\includegraphics[width=70px]{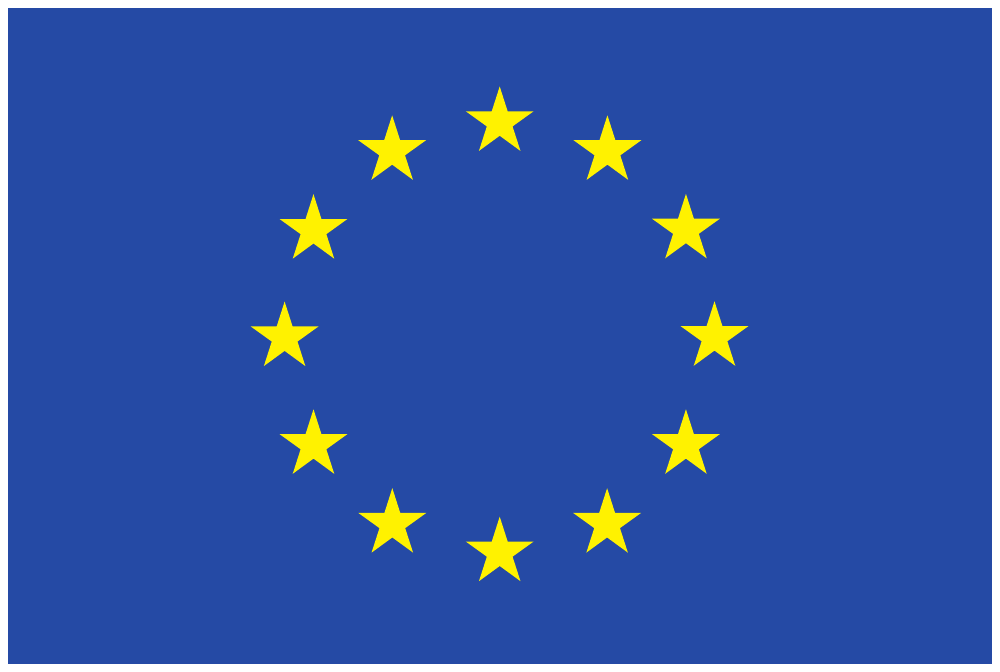}%
\end{textblock}

\begin{abstract}
  It is a long-standing open problem whether the minimal dominating sets of a graph can be enumerated in output-polynomial time. 
  In this paper we investigate this problem in graph classes defined by forbidding an induced subgraph. 
  In particular, we provide output-polynomial time algorithms for $K_t$-free graphs and for several related graph classes. 
  This answers a question of Kanté et al.~about enumeration in bipartite graphs.
\end{abstract}

\newcommand{\figscaleAx}{1}
\newcommand{\figscaleAy}{1}
\newcommand{\figscaleBx}{1}
\newcommand{\figscaleBy}{.85}
\newcommand{\figscaleCx}{0.9}
\newcommand{\figscaleCy}{1}


\section{Introduction}\label{sec:intro}

Countless algorithmic problems in graph theory require to detect a structure with prescribed properties in an input graph.
Rather than finding one such object, it is sometimes more desirable to generate all of them. 
This is for instance useful in certain applications to database search \cite{yan2005substructure}, network analysis \cite{grochow2007network}, bioinformatics \cite{damaschke2006parameterized,marino2015analysis}, and cheminformatics~\cite{barnard1993substructure}.
Enumeration algorithms for graph problems seem to have been first mentioned in the early 70's with the pioneer works of Tiernen \cite{tiernan1970efficient} and Tarjan \cite{tarjan1973enumeration} on cycles in directed graphs, and of Akkoyunlu \cite{akkoyunlu1973enumeration} on maximal cliques in undirected graphs. 
However, they already appeared in disguise in earlier works \cite{paull1959minimizing, marcus1964derivation}.  
To this date, several intriguing questions on the topic remain unsolved.  
We refer the reader to \cite[Chapter~2]{marino2015analysis} and~\cite{strozecki2019survey} for more in-depth introductions to enumeration algorithms, and to \cite{wasa2016enumeration} for a listing of enumeration algorithms and problems.

The objects we wish to enumerate in this paper are the
(inclusion-wise) minimal dominating sets of a given graph. In general,
the number of these objects may grow exponentially with the order $n$
of the input graph.
Therefore, in stark contrast to decision or optimization problems, looking for a running time polynomially bounded by $n$ is not a reasonable, let alone meaningful, efficiency criterion. 
Rather, we aim here for so-called \emph{output-polynomial} time algorithms~\cite{johnson1988generating},
whose running time is polynomially bounded by the size of both the input
and output data, and refer to~\cite{fomin2008combinatorial,couturier2013minimal,golovach2019input} for input exponential-time algorithms for the problem we consider in this paper.

Since dominating sets are among the most studied objects in graph
theory and algorithms, their enumeration (and counting) have attracted an increasing
attention over the past 10 years.
The problem of enumerating minimal dominating sets (hereafter referred to as \DomEnum{}) has a notable feature: 
it is equivalent to the extensively studied problem~\transv{}. 
In \transv{}, one is given a hypergraph $\mathcal{H}$ (\ie,~a collection of subsets, called \emph{hyperedges}, of elements called \emph{vertices}) and is asked to enumerate all the (inclusion-wise) minimal
\emph{transversals} of $\mathcal{H}$ (\ie, the inclusion-wise minimal
sets of vertices that meet every hyperedge). It is not hard to see
that \DomEnum{} is a particular case of \transv{}: the minimal
dominating sets of a graph $G$ are exactly the minimal transversals of
the hypergraph
of closed neighborhoods of $G$. 
Conversely, Kanté,
Limouzy, Mary, and Nourine proved that every instance of \transv{} can
be reduced to a co-bipartite\footnote{The complement of a bipartite graph.} instance of \DomEnum{}
\cite{kante2014split}. 
Currently, the best output-sensitive
algorithm for \transv{} is due to Fredman and Khachiyan and runs in
output quasi-polynomial time~\cite{fredman1996complexity}. It is a
long-standing open problem whether this complexity bound can be
improved (see for instance the surveys \cite{eiter2002hypergraph, eiter2008computational}).
Therefore, the equivalence between the two problems is an additional
motivation to study \DomEnum, with the hope that techniques from graph
theory could be used to obtain new results on \transv{}.

So far, output-polynomial time algorithms have been obtained for
\DomEnum{} in several classes of graphs, including planar graphs and degenerate graphs \cite{eiter2003new}, classes of graphs of bounded treewidth, cliquewidth
\cite{courcelle2009linear}, or LMIM-width \cite{golovach2018lmimwidth}, path graphs and line graphs \cite{kante2012neighbourhood}, interval graphs and permutation graphs
\cite{kante2013permutation}, split graphs
\cite{kante2015chordal}, graphs of girth at least~7~\cite{golovach2015flipping}, chordal
graphs \cite{kante2015chordal}, and chordal bipartite
graphs~\cite{golovach2016chordalbip}. A succinct survey of results on
\DomEnum{} can be found in \cite{kante2008encyclopedia} and~\cite{golovach2016chordalbip}.

In this paper, we investigate the complexity of \DomEnum{} in graph classes defined by forbidding an induced subgraph $H$, hereafter referred to as \emph{$H$-free} graphs.
For every $t\in \N$, we denote by $K_t$ the complete graph on $t$ vertices, by $K_t-e$ the graph obtained by removing any edge in~$K_t$ and by $K_t+K_2$ the disjoint union of $K_t$ and $K_2$.
Our main result is the following.

\begin{theorem}\label{thm:op}
There is an algorithm enumerating, for every fixed $t \in \N$, the minimal dominating sets in $(K_t+K_2)$-free graphs in output-polynomial time and polynomial space.
\end{theorem}

In particular, this yields an output-polynomial time algorithm for $K_t$-free graphs. 
A~notable special case is that of bipartite graphs, where the question of the existence of an output-polynomial time algorithm for \DomEnum{} was explicitly stated in~\cite{kante2008encyclopedia} and later papers~\cite{kante2015chordal,golovach2016chordalbip}.
We stress that we provide in the proof of Theorem~\ref{thm:op} a single algorithm that deals with all values of~$t$ and that this algorithm does not require the knowledge of~$t$. We discuss the complexity in greater details in Sections~\ref{sec:triangle-free} and~\ref{sec:kt-free}.

In order to push our techniques to their limits, we investigate cases that are close to but not covered by Theorem~\ref{thm:op}. Namely, we consider two particular choices of the graph $H$: the \emph{paw}, which is the graph obtained by adding a vertex of degree one to $K_3$
(\ie{}, $H = \tikz[every node/.style = small node, scale = 0.2, baseline=-0.1cm]{
\draw (0:1) node (a) {} -- (120:1) node {} -- (-120:1) node {} -- (a) -- ++ (1.4759, 0) node {};
}$), and $K_4-e$, also known as \emph{diamond} graph
(%
\ie{}, $H = \tikz[every node/.style = small node, scale = 0.2, baseline=-0.1cm]{
\draw (0:1) node (a) {} -- (120:1) node (b) {} -- (-120:1) node (c) {} -- (a) ++ (-2.95189, 0) node (d) {} (b) -- (d) -- (c);
}$). We combine our main tools with some ad hoc techniques to handle those two cases, thus obtaining the following.

\begin{theorem}\label{th:pawdiam}
There is an algorithm enumerating minimal dominating sets in paw-free (resp. diamond-free) graphs in output-polynomial time and polynomial space.
\end{theorem}

Our algorithms first decompose the input graph by successively removing closed neighborhoods in the fashion of~\cite{eiter2003new}.
We then follow this decomposition to construct partial minimal dominating sets, adding the neighborhoods back one after the other. A crucial point of this approach, known as \emph{ordered generation}, is that we can relate the enumeration of potential extensions of a partial minimal dominating set to the \DomEnum{} problem in a simpler class.

The paper is organized as follows. In Section~\ref{sec:prelim} we give the necessary definitions. 
The graph decompositions that we use, called \emph{peelings}, are introduced in Section~\ref{sec:bt} along with their main properties.
In Section~\ref{sec:triangle-free}, we give an algorithm for \DomEnum{} that runs in output-polynomial time in triangle-free graphs with better time bound than that coming from Theorem~\ref{thm:op}.
A generalization of this algorithm for $K_t$-free graphs is given in Section~\ref{sec:kt-free} (Theorem~\ref{thm:kt-free}). 
This algorithm is then extended to $(K_t+K_2)$-free graphs in Section~\ref{sec:variants} (Theorem~\ref{thm:ktme}).
In the same section, algorithms are given for diamond-free graphs (Theorem~\ref{thm:diamond-free}) and paw-free graphs (Theorem~\ref{thm:paw-free}), \ie, the two cases of Theorem~\ref{th:pawdiam}. 
We discuss in Section~\ref{sec:beyond} the obstacles to stronger theorems using the same tools. 
Finally, we conclude with possible future research directions in Section~\ref{sec:concl}.

\section{Preliminaries}
\label{sec:prelim}

\paragraph{Graphs.}
All graphs in this paper are finite, undirected, simple, and loopless. 
If $G=(V(G),E(G))$ is a graph, then $V(G)$ is its set of vertices and 
$E(G)\subseteq \{\{x,y\} \mid x,y\in V(G),\ x\neq y\}$ is its set of edges.
Edges are denoted by $xy$ (or $yx$) instead of $\{x,y\}$.
We assume that vertices are assigned distinct indices; these will be used to choose vertices in a deterministic way, typically selecting the vertex of smallest index.
A {\em clique} (respectively an \emph{independent set}) in a graph $G$ is a set of pairwise adjacent (respectively non-adjacent) vertices.
We note $\omega(G)$ the size of a largest clique in $G$.
The subgraph of $G$ \emph{induced} by $X\subseteq V(G)$, denoted by $G[X]$, is the graph $(X,E(G)\cap \{\{x,y\} \mid x,y\in X,\ x\neq y\})$; $G-X$ is the graph $G[V(G)\setminus X]$.
For every graph $H$, we say that $G$ is $H$-free if no induced subgraph of $G$ is isomorphic to~$H$.
If a vertex $v \in V(G)$ is adjacent to every vertex of a set $S \subseteq V(G)$, we say that $v$ is \emph{complete} to~$S$.

If the vertex set of a graph $G$ can be partitioned into one part inducing a clique and one part inducing an independent set (respectively two independent sets, two cliques), we say that $G$ is a \emph{split} (respectively \emph{bipartite}, \emph{co-bipartite}) graph.
If $f$ is a function, we write $f(n) = \poly(n)$ when there is a constant $c\in \mathbb{N}$ such that $f(n) \in O(n^c)$.

\paragraph{Neighbors and domination.}
Let $G$ be a graph and $x \in V(G)$.
We write $N(x)$ for the set of {\em neighbors} of $x$ in $G$ defined by $N(x)=\{y\in V(G)\mid xy\in E(G)\}$; $N[x]$ is the {\em closed neighborhood} of $x$ defined by $N[x]= N(x)\cup\{x\}$.
For a given $X\subseteq V(G)$, we note $N[X]=\bigcup_{x\in X} N[x]$ and $N(X)=N[X]\setminus X$.
Let $D$ be a set of vertices of $G$.
We say that $D$ \emph{dominates} a subset $S \subseteq V(G)$ if $S \subseteq N[D]$. It \emph{minimally dominates} $S$ if no proper subset of $D$ dominates $S$.
The set $D$ is a (\emph{minimal}) \emph{dominating set} of $G$ if it (minimally) dominates $V(G)$.
The set of all minimal dominating sets of $G$ is denoted by $\D(G)$ and the problem of enumerating $\D(G)$ given $G$ is denoted by \DomEnum{}.

Let $S \subseteq V(G)$.
A vertex $y \in V(G)$ is said to be a \emph{private neighbor} of some $x \in S$ if it is only dominated by $x$ in $S$, \ie, if $y\in N[S]$ but $y\not\in N[S\setminus \{x\}]$. Note that $x$ can be its own private neighbor.
The set of private neighbors of $x\in S$ in $G$ is denoted by $\priv_G(S,x)$ and we drop the subscript when it can be inferred from the context.
Observe that $S$ is a minimal dominating set of $G$ if and only if $V(G) \subseteq N[S]$ and for every $x\in S$, $\priv(S,x)\neq \emptyset$.

\paragraph{Enumeration.}
In this paper, to measure time and space complexity we assume the RAM model, where any integer can be stored in a single register and arithmetic operations on integers have unit cost~\cite{strozecki2019survey}. 
The aim of graph enumeration algorithms is to generate a set $\mathcal{X}(G)$ of objects related to a graph $G$. 
We say that an algorithm enumerating $\mathcal{X}(G)$ on input being an $n$-vertex graph $G$ is running in \emph{output-polynomial} time if its running time is polynomially bounded by the sizes of the input and output data, \ie, $n + |\mathcal{X}(G)|$. If an algorithm enumerates $\mathcal{X}(G)$ by spending $\poly(n)$-time (respectively $O(n)$-time) before it outputs the first element, between two output elements, and after it outputs the last element, then we say that it runs with \emph{polynomial delay} (respectively \emph{linear delay}). It is easy to see that every polynomial delay algorithm is also output-polynomial. Note, however, that there exist problems that admit output-polynomial time algorithms but no polynomial delay ones, unless \TFNPFP{} \cite{strozecki2019survey}. When discussing the space used by an enumeration algorithm, we mean the working space and we ignore the space where the solutions are output. If the existence of an output-polynomial time algorithm for a problem implies the existence of one for \DomEnum{}, we say that this problem is \DomEnum{}-hard. As mentioned in the introduction, we have the following. 

\begin{theorem}[Kant\'e et al.~\cite{kante2014split}]\label{thm:cobip-hard}
    \DomEnum{} restricted to co-bipartite graphs is \DomEnum{}-hard.
\end{theorem}

\smallskip

\section{Ordered generation in bicolored graphs}\label{sec:bt}

In this section, we give a general procedure that will be used in the rest of this paper for the enumeration of minimal dominating sets in \mbox{$K_t$-free} graphs and in related graph classes.
This procedure will construct minimal dominating sets one neighborhood at a time, in a variant of what is known as the {\em backtrack search technique} in \cite{read1975bounds,fukuda1997analysis,strozecki2019efficient}, and referred to as {\em ordered generation} in \cite{eiter2003new}.

In what follows, we find it more convenient to deal with the slightly more general setting of domination in bicolored graphs.
A \emph{bicolored graph} is a graph together with a subset of its vertex set. 
For a graph $G$ and a subset $A \subseteq V(G)$, we denote by $G(A)$ the bicolored graph $G$ with prescribed set~$A$.
We also say that $G$ has \emph{bicoloring} $(A,V(G)\setminus A)$.
Then, a \emph{dominating set of $G(A)$} is a set $D \subseteq V(G)$ that dominates $A$, \ie, such that $A \subseteq N[D]$.
It is (inclusion-wise) minimal if it does not contain any dominating set of $G(A)$ as a proper subset.
Intuitively, the vertices of $G-A$ may be used in the dominating set, but do not need to be dominated.
For every graph $G$ and subset $A\subseteq V(G)$, we denote by $\D(G, A)$ the set of minimal dominating sets of~$G(A)$.
Then $\D(G,A)=\D(G)$ whenever $A=V(G)$.

A \emph{peeling} of a bicolored graph $G(A)$ is a sequence of vertex sets $(V_0, \dots, V_p)$ such that $V_p = A$, $V_0 = \emptyset$, and for every $i\in \intv{1}{p}$, there is a vertex $v_i \in V_i$ such that
\[
  V_{i-1} = V_i \setminus N[v_i].
\]
We call $(v_1, \dots, v_p)$ the \emph{vertex sequence} of the peeling.
It is straightforward to see that given a bicolored graph $G(A)$, any peeling of $G(A)$ can be computed in $O(n^2)$ time and space: start with the whole set $A$, and as long as $A$ remains non-empty, pick a vertex $v$ in it and
remove $N[v]$ from $A$.
The representation of a peeling is given in Figure~\ref{fig:peel}. 

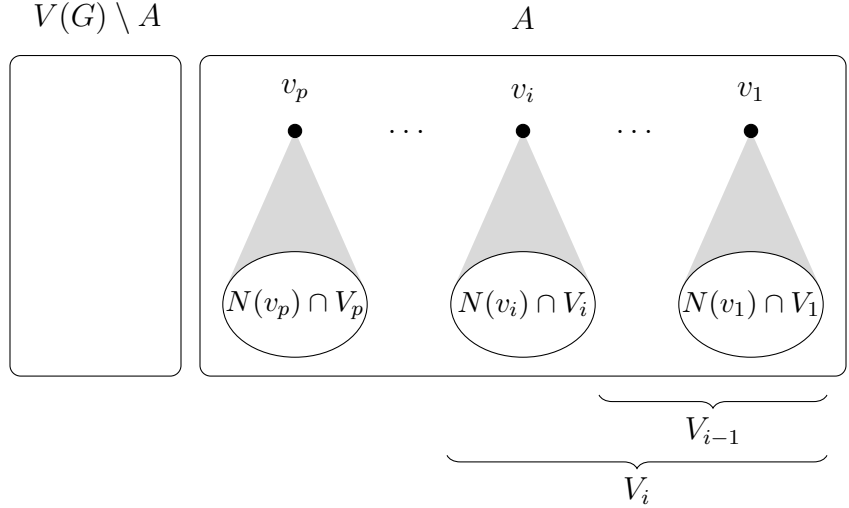
\begin{figure}
  \centering
  \begin{tikzpicture}[xscale = \figscaleAx, yscale = \figscaleAy]
    \draw[rounded corners] (-0.25,0) rectangle (8.25,-4.25);
    \draw (4, 0.5) node[normal] {$A$};
    \draw[rounded corners] (-0.5,0) rectangle (-2.75,-4.25);
    \draw (-1.6, 0.5) node[normal] {$V(G)\setminus A$};
    %
    \draw[every node/.style = black node, label distance=2mm]
    (1, -1) node[label = 90:$v_p$] (vp) {}
    (4, -1) node[label = 90:$v_{i}$] (vi) {}
    (7, -1) node[label = 90:$v_{1}$] (v1) {};
    \draw (2.5, -1) node[normal]{$\dots$};
    \draw (5.5, -1) node[normal]{$\dots$};
    %
    \fill[fill = black, opacity = 0.15] (vp) -- ++(-0.92cm, -2.1) -- ++(2*0.92, 0) -- (vp);
    \fill[fill = black, opacity = 0.15] (vi) -- ++(-0.92cm, -2.1) -- ++(2*0.92, 0) -- (vi);
    \fill[fill = black, opacity = 0.15] (v1) -- ++(-0.92cm, -2.1) -- ++(2*0.92, 0) -- (v1);
    %
    \draw[fill = white]
    (1, -3.3) ellipse (0.95cm and 0.7cm) node[normal] {\small $N(v_p)\cap V_p$}
    (4, -3.3) ellipse (0.95cm and 0.7cm) node[normal] {\small $N(v_{i})\cap V_i$}
    (7, -3.3) ellipse (0.95cm and 0.7cm) node[normal] {\small $N(v_1)\cap V_1$};
    %
    \draw[decorate,decoration={brace,amplitude=5pt}] (8, -4.5) -- ++(-3, 0) node[midway, yshift = -0.15cm, anchor = north, normal] {$V_{i-1}$};
    \draw[decorate,decoration={brace,amplitude=5pt}] (8, -5.3) -- ++(-5, 0) node[midway, yshift = -0.15cm, anchor = north, normal] {$V_i$};
  \end{tikzpicture}
  \caption{Representation of a peeling of a bicolored graph $G(A)$ constructed by iteratively removing $v_i$'s and their neighborhoods, for $i$ from $p$ to $1$.
  Note that vertices that are effectively removed at step $i$ are those of $N[v_i]\cap V_i$, as vertices in $N[v_i]\setminus V_i$ have already been removed at a previous step.
  A crucial property is that $v_i$ has no neighbor in $V_{i-1}$.}
  \label{fig:peel}
\end{figure}

In the remaining of this section, 
we consider a bicolored graph $G(A)$, together with a fixed peeling
$(V_0, \dots, V_{p})$ of $G(A)$ with vertex sequence $(v_1,\dots,v_p)$. 
Observe that $\D(G,V_p)=\D(G,A)$.
We now define the relation that will be used by our algorithm to enumerate the minimal dominating sets of $G(A)$ without repetition.
Recall that the sets of $\D(G, V_i)$ may contain vertices of $G-V_i$, which is a crucial point.

\begin{definition}\label{def:parent}
    Let $i\in \intv{0}{p-1}$ and $D\in \D(G, V_{i+1})$. We define $\parent(D,i+1)$ as the pair $(D^*,i)$ where $D^*$ is obtained from $D$ by exhaustively applying the following operation: as long as there exists a vertex $x$ in $D$ satisfying $\priv(D,x)\cap V_i=\emptyset$, remove from $D$ the vertex of smallest index with this property. 
\end{definition}

Clearly, there is a unique way to build $\parent(D,i+1)$ given $D$ and $i$. 
By construction, the obtained set $D^*$ is a minimal dominating set of $G(V_i)$. 
Hence, every set in $\D(G,V_{i+1})$ can be obtained by completing some $D^*$ in $\D(G,V_i)$; we develop this point below.

\needspace{0.4in}
\begin{proposition}\label{prop:parent-vi}
    Let $i\in \intv{0}{p-1}$ and $D^*\in \D(G,V_i)$. Then:
    \begin{enumerate}[(i)]
        \item if $D^*$ dominates $V_{i+1}$ then $D^* \in \D(G,V_{i+1})$ and $\parent(D^*,i+1)=(D^*,i)$;
        \item otherwise, $D^*\cup \{v_{i+1}\} \in \D(G,V_{i+1})$ and $\parent(D^*\cup \{v_{i+1}\},{i+1})=(D^*,i)$.
    \end{enumerate}
\end{proposition}

\begin{proof}
    First note that since $D^*\in \D(G,V_i)$, for all $x\in D^*$ we have $\priv(D^*,x)\cap V_i\neq \emptyset$, implying also $\priv(D^*,x)\cap V_{i+1}\neq \emptyset$. Hence, if $D^*$ dominates $V_{i+1}$ then in fact $D^*$ is a minimal dominating set of $G(V_{i+1})$ and thus $D^*\in \D(G,V_{i+1})$. Since $\priv(D^*,x)\cap V_i\neq \emptyset$ for all $x\in D^*$, we then have that $\parent(D^*,i+1)=(D^*,i)$ directly from the definition.
    
    Suppose now that $D^*$ does not dominate $V_{i+1}$, and observe that then $D=D^*\cup \{v_{i+1}\}$ does.
    Moreover, $\priv(D,v_{i+1})\cap V_{i+1}\neq\emptyset$.
    Since $v_{i+1}$, by the definition of the peeling, is not adjacent to any vertex in $V_i$, it cannot steal any private neighbor from the elements of~$D^*$, \ie, $\priv(D^*,x)\cap V_i\neq\emptyset$ implies $\priv(D^*\cup \{v_{i+1}\},x)\cap V_i\neq\emptyset$ for any $x\in D^*$.
    Hence $\priv(D,x)\cap V_{i+1}\neq \emptyset$ for all $x\in D$. 
    Now, note that since $v_{i+1}$ does not steal private neighbors from the elements of~$D^*$, it is indeed the only node in $D$ with no private neighbors in $V_i$, and it is removed when constructing $\parent(D,i+1)$.
    Hence $\parent(D,i+1)=(D^*,i)$, as claimed.
\end{proof}

The $\parent$ function as introduced in Definition~\ref{def:parent} defines a tree on vertex set 
\[\{(D,i) \mid i \in \intv{0}{p},\ D \in \D(G,V_i)\},\] 
with leaves $\{(D,p) \mid D \in \D(G, A)\}$ and root $(\emptyset,0)$ (the empty set being the only minimal dominating set of the empty vertex set $V_0$). 
Our algorithms will search this tree in order to enumerate the minimal dominating sets of~$G(A)$.
Proposition~\ref{prop:parent-vi} guarantees that for every $i<p$ and every $D^* \in \D(G,V_i)$, the pair $(D^*,i)$ is the parent of some $(D,i+1)$ with $D \in \D(G,V_{i+1})$ (possibly $D = D^*$). 
Consequently, every branch of the tree leads to a different minimal dominating set of $G(A)$. 
In particular, for every $i \in \intv{0}{p-1}$ we have
\begin{equation}\label{eq:depthibound}
    |\D(G,V_i)| \leq |\D(G,V_{i+1})| \leq |\D(G,A)|.
\end{equation}

Given a set $D^* \in \D(G, V_i)$, we now focus on the enumeration of all $D \in \D(G,V_{i+1})$ such that $\parent(D, i+1)=(D^*, i)$. 
Any (inclusion-wise) minimal set $X\subseteq V(G)$ such that $V_{i+1}\subseteq N[D^*\cup X]$ will be called a {\em candidate extension} of $(D^*, i)$.
In~other words, $X$ is a candidate extension of $(D^*,i)$ if and only if it is a minimal dominating set of the bicolored graph $G$ with prescribed set $V_{i+1}\setminus N[D^*]$.
Then, we denote by $\C(D^*,i)$ the set of all candidate extensions of~$(D^*,i)$, \ie,
\begin{equation}\label{eq:cdstari}
        \C(D^*,i)\eqdef \D(G,V_{i+1}\setminus N[D^*]).
\end{equation}
Observe that if $(D,i+1)$ has $(D^*, i)$ as its parent, then $D\setminus D^*$ is candidate extension of $(D^*,i)$.
From Proposition~\ref{prop:parent-vi}, we also know that one of $(D^*, i + 1)$ and $(D^*\cup\{v_{i+1}\}, i + 1)$ has $(D^*, i)$ as its parent, hence that either $\emptyset$ or $\{v_{i+1}\}$ is a candidate extension of $(D^*,i)$.
Note that we have no guarantee that any other candidate extension forms a minimal dominating set of $V_{i+1}$, together with $D^*$.
We show that it is still reasonable to test each of the candidate extensions even though $D^*$ might have a unique child.

\begin{lemma}\label{lem:cand-ext-bound}
    Let $H(B)$ be a bicolored graph and $D \subseteq V(H)$. Then
    \[
        |\D(H, B \setminus N[D])| \leq |\D(H,B)|.
    \]
\end{lemma}

\begin{proof}
We argue that for every $X \in \D(H, B \setminus N[D])$ we can find a set $D_X\in \D(H, B)$ so that the sets $D_X$ are pairwise different for different $X$; this assertion immediately implies the desired inequality.
For this, we define $D_X$ as any minimal dominating set of $H(B)$ that is a subset of $D \cup X$; such a set exists as $D \cup X$ dominates~$B$.
By definition, every vertex of $X$ has a private neighbor in $B \setminus N[D]$ so we have $X \subseteq D_X$. Moreover, since $X$ is a minimal dominating set of $B \setminus N[D]$, $X$ is disjoint with $D$. We conclude that $X=D_X\setminus D$, and hence that the sets $D_X$ are pairwise different for different $X$.
\end{proof}

As a consequence of Lemma~\ref{lem:cand-ext-bound} and Inequality~\eqref{eq:depthibound}, we have the following.
\begin{corollary}\label{cor:cibound}
Let $i\in \intv{0}{p-1}$ and $D^*\in \D(G,V_i)$.
Then $|\C(D^*,i)|\leq|\D(G, A)|$.
\end{corollary}

We conclude the ordered generation procedure with the following statement, which reduces the existence of an output-polynomial time algorithm enumerating $\D(G,A)$ to the existence of one enumerating $\C(D^*,i)$ for any $i\in \intv{0}{p-1}$ and $D^*\in \D(G,V_i)$.

\begin{theorem}\label{thm:ordered-generation}
    Let $f\colon \N^2 \to \N$ and $s\colon \N\to \N$ be two functions.
    Assume that there is an algorithm that, given a bicolored graph $G(A)$ on $n$ vertices, 
    a peeling $(V_0,\dots, V_{p})$ of $G(A)$, $i\in \intv{0}{p-1}$, and $D^*\in \D(G,V_i)$, enumerates $\C(D^*, i)$ in time at most $f(n,|\D(G,A)|)$ and space at most $s(n)$.
    Then there is an algorithm that, given a bicolored graph $G(A)$ on $n$ vertices, enumerates the set $\D(G,A)$ in time
    \[
        O(n^4d^2 + f(n,d)\cdot nd)
    \]
    and space $O(n \cdot s(n))$, where $d=|\D(G,A)|$.
\end{theorem}

\begin{proof}
    Let us assume that there exists an algorithm \algob{} that, given a bicolored graph $G(A)$ on $n$ vertices, a peeling $(V_0,\dots, V_{p})$ of $G(A)$, $i\in \intv{0}{p-1}$ and $D^*\in \D(G,V_i)$, enumerates $\C(D^*,i)$ in time at most $f(n,|\D(G,A)|)$ and space at most $s(n)$. Note that we may assume that $s(n)\in \Omega(n)$, as \algob{} needs to store its input.
    We describe an algorithm \algoa{} that enumerates $\D(G,A)$ within the specified time and space complexities.
    
    The algorithm first checks whether $A = \emptyset$ and, if so, returns $\{\emptyset\}$.
    Otherwise, it computes a peeling $(V_0,\dots,V_{p})$ of $G(A)$ in time $O(n^2)$ and using $O(n^2)$ space.
    Recall that the $\parent$ relation defines a tree $T$ on vertex set 
    \[
        \{(D,i) \mid i \in \intv{0}{p},\ D \in \D(G,V_i)\},
    \] 
    with leaves $\{(D,p) \mid D \in \D(G, A)\}$ and root $(\emptyset,0)$. 
    Therefore, in order to enumerate $\D(G,A)$, it is enough for \algoa{} to enumerate the leaves of~$T$.
    To do so, the algorithm performs a depth-first search (DFS) of $T$ outputting each visited leaf.
    For each node $(D^*,i)$, $i\in \intv{0}{p-1}$ of $T$, the algorithm runs \algob{} on input $(G(A), (V_0,\dots,V_{p}), i, D^*)$ to generate $\C(D^*,i)$ in time $f(n,d)$ and space $s(n)$.
    For every $X\in \C(D^*,i)$ generated by~\algob{}, the algorithm tests whether $D^*\cup X$ is a minimal dominating set of $V_{i+1}$, and whether $\parent(D^*\cup X,i+1)=(D^*,i)$.
    This requires $O(n^3)$ steps per candidate extension, and a total working space of $O(n)$, disregarding the space needed to store the (globally fixed) graph~$G$.
    As by Corollary~\ref{cor:cibound} we have $|\C(D^*,i)|\leq |\D(G,A)|=d$, the total time spent by \algoa{} at each node of $T$ is bounded
    by $O(n^3d + f(n,d))$.
    By Inequality~\eqref{eq:depthibound} we have $|V(T)|\leq pd$ and clearly $p\leq n$, so the total running time of \algoa{} is bounded by 
    \[
        O(n^4d^2 + f(n,d)\cdot nd).
    \]
    
    Regarding the space, we observe that whenever we visit a node of $T$, we do not need to compute the whole set of its children. Instead, it is enough in order to continue the DFS to compute the next unvisited child only, which can be done using \algob{} and pausing it afterward. Therefore, when we visit some $(D, i) \in V(T)$, we only need to store the data of the $i-1$ (paused) executions of \algob{} enumerating the children of the ancestors of $(D, i)$, plus the data of the algorithm enumerating the children of $D$, \ie, $i\cdot (O(n)+s(n))$ space. 
    As $s(n) \in \Omega(n)$, the described algorithm uses $O(n\cdot s(n))$ space, as claimed.
\end{proof}

\section{Candidate extensions in triangle-free graphs}\label{sec:triangle-free}

We show that candidate extensions can be enumerated in output-polynomial time in triangle-free graphs, which by Theorem~\ref{thm:ordered-generation} leads to an output-polynomial time algorithm enumerating minimal dominating sets in this class of graphs.
In fact, our result holds in the more general context where only the graph induced by the set that needs to be dominated is required to be triangle-free, and not necessarily the whole graph, a point that is discussed in Section~\ref{sec:beyond}.

In the following, we consider a bicolored graph $G(A)$ on $n$ vertices.
Moreover, we have a fixed peeling $(V_0,\dots, V_{p})$ of $G(A)$ with vertex sequence $(v_1,\dots,v_p)$.
Then, we consider
\begin{align*}
    & i \in \intv{0}{p-1},
    & D^* \in \D(G,V_i),
\end{align*}
and define $\C(D^*,i)$ as in Equality~\eqref{eq:cdstari} in Section~\ref{sec:bt}.
We will show how to enumerate $\C(D^*,i)$ in output-polynomial time whenever $G[A]$ is triangle-free.

Kant\'e, Limouzy, Mary, and Nourine gave the following characterization of minimal dominating sets in split graphs.

\begin{proposition}[\cite{kante2014split}]\label{prop:split-properties}
    Let $H$ be a split graph with vertices partitioned into an independent set $S$ and a clique~$C$, where $S$ is taken to be (inclusion-wise) maximal.
    Then, for every $D\in \D(H)$ the following holds:
    \begin{enumerate}[(i)]
        \item $D\cap S=S\setminus N(D\cap C)$, so in particular $D$ is uniquely determined by its intersection with the clique; and
        \item for every $x\in D$, $\priv(D,x)\cap S\neq \emptyset$.
    \end{enumerate}
    Furthermore, $\D(H)$ can be enumerated with delay $O(n^2)$ and using $O(n^2)$ space.
\end{proposition}

We can now use Proposition~\ref{prop:split-properties} to establish the following understanding of candidate extensions in terms of minimal dominating sets of an auxiliary split graph.
The set $S$ in the next lemma corresponds to the elements that, together with $v_{i+1}$, must be dominated by the candidate extensions of $(D^*,i)$.
This situation is depicted in Figure~\ref{fig:characterization}.

\begin{lemma}\label{lem:triangle-free-candidates-characterization}
    Suppose that $N(v_{i+1})$ is an independent set.
    Let $S = V_{i+1}\setminus (N[D^*] \cup \{v_{i+1}\})$, $C=N(S)\setminus \{v_{i+1}\}$, and let $H$ be the split graph obtained from $G[S\cup C]$ by completing $C$ into a clique.
    Then, every set $X$ in $\C(D^*,i)$ either belongs to $\D(H)$, or belongs to $\D(H)$ after removing one vertex (\ie, $X\setminus \{u\}\in \D(H)$ for some $u\in X$), or is such that $X=\{v_{i+1}\}$.
	Moreover, $|\D(H)|\leq n \cdot |\C(D^*,i)|+1$.
\end{lemma}
\begin{proof}
Consider any $X\in \C(D^*,i)$, $X\neq \{v_{i+1}\}$. Then, by definition, $X$ is a minimal dominating set of $V_{i+1}\setminus N[D^*]$. 
By assumption as $V_{i+1}\setminus N[D^*]\subseteq N[v_{i+1}]$, we have $v_{i+1}\notin X$.
Observe then that $X\subseteq C\cup S$.

We first consider the case when $v_{i+1}\in N[D^*]$. Then $V_{i+1}\setminus N[D^*]=S$ and, as $H$ is a supergraph of $G[S\cup C]$ and $S$ remains an independent set in $H$, it follows that $X$ minimally dominates $S$ in $H$. Note that either $X$ contains a vertex of $C$, and then this vertex dominates $C$ in $H$, or $X=S$, and then $X$ dominates $C$ in $H$ as well. We conclude that $X$ is a minimal dominating set of $H$, \ie{}, that $X\in \D(H)$ in this case.

We now consider the remaining case when $v_{i+1}\notin N[D^*]$; then $V_{i+1}\setminus N[D^*]=S\cup \{v_{i+1}\}$. Observe that now either $X$ is a minimal dominating set of $S$, or there exists $u\in N(v_{i+1})\cap X$ such that $X\setminus \{u\}$ is a dominating set of $S\setminus N[u]$, \ie{}, of $S\setminus \{u\}$ as $N(v_{i+1})$ is an independent set.
Denote $D=X$ in the former case and $D=X\setminus \{u\}$ in the latter case; we now apply a reasoning similar to that from the previous paragraph.
Since $v_{i+1}\notin D$ and $D$ is a minimal dominating set of $S$ or $S \setminus \{u\}$, it follows that $D\subseteq C\cup S$.
If $D\subseteq S$, then either $D = S \setminus \{u\}$ (if $u$ is defined and $u \in S$) or $D = S$ (otherwise), because $N(v_{i+1})$ (hence in particular $S$) is an independent set.
If $D\not\subseteq S$, then $D\cap C\neq\emptyset$.
In both cases, $D$ dominates $C$ in~$H$, and we conclude that $X$ (if $u$ is defined and $u \in S$) or $D$ (otherwise) is a dominating set of~$H$. Moreover, since $\priv(D,x)\cap S\neq\emptyset$ for each $x\in D$ it follows that $X$ or $D$ is a minimal dominating set of $H$. As $X=D$ or $X=D\cup \{u\}$ for some vertex $u$, we conclude that either $X$ belongs to $\D(H)$, or it belongs to $\D(H)$ after removing one vertex, a vertex that was here only to dominate $v_{i+1}$.

Having considered both cases, the claimed property of elements of $\C(D^*,i)$ follows.
We are left with proving the claimed upper bound on $|\D(H)|$. We first show that 
\begin{align}
    |\{D\in \D(H) \mid D\cap S=\emptyset\}| \leq (n-1) \cdot |\{D\in \D(H) \mid D\cap S\neq \emptyset\}|+1.\label{eq:triangle-free-trashbound}
\end{align}
Indeed, consider the map $f$ that, given $D \in \{D\in \D(H) \mid D\cap S=\emptyset\}$, $D\neq \emptyset$, removes one arbitrary vertex from $D$, and completes the dominating set by adding all the vertices in the independent set which are no longer dominated.
Then, $f$ maps non-empty elements of $\{D\in \D(H) \mid D\cap S=\emptyset\}$ to the set $\{D\in \D(H) \mid D\cap S\neq\emptyset\}$. 
Moreover, every element in this second set is the image of at most $|C| \leq n-1$ elements by $f$. 
This implies the desired bound.

From Inequality~\eqref{eq:triangle-free-trashbound} we immediately obtain that
$$|\D(H)| \leq n \cdot |\{D\in \D(H) \mid D\cap S\neq \emptyset\}|+1,$$
so it suffices to prove that 
$$|\{D\in \D(H) \mid D\cap S\neq \emptyset\}|\leq |\C(D^*,i)|.$$
To see this, we observe that in fact we have $\{D\in \D(H) \mid D\cap S\neq \emptyset\}\subseteq \C(D^*,i)$.
Indeed, by Proposition~\ref{prop:split-properties} we have that every $D\in \D(H)$ is a minimal dominating set of $S$ in $G$, and it moreover dominates $v_{i+1}$ provided $D\cap S\neq \emptyset$.
\end{proof}


\definecolor{redcol}{HTML}{A9A9A9}
\definecolor{bluecol}{HTML}{FFFFFF}
\definecolor{greencol}{HTML}{FFFFFF}

\begin{figure}
    \centering
    \begin{tikzpicture}[every node/.style = black node, xscale=\figscaleBx, yscale=\figscaleBy]
    \draw[rounded corners, dotted, thick, fill = redcol] (2.5, 1.25) rectangle (4.5, -3.5) node[normal, label={[label distance = 1.5mm]135:$V_{i}$}] {};
    \draw[rounded corners, dotted, thick] (-2, 1.5) node[normal, label={[label distance = 1.5mm]-45:$V_{i+1}$}] {} rectangle (4.75, -3.75);
    
    \draw
    (0, 0.25) node[black node, label = 90:$v_{i+1}$] (v) {};
    
    \fill[fill = black, opacity = 0.08] (v) -- ++(-1.58cm, -2.5) -- ++(2*1.58, 0) -- (v);
    \fill[fill = black, opacity = 0.08] (3.5, 1.25) -- (1.618, -2.5) -- (0, -1.5) -- (1.5, 1.25) -- (3.5, 1.25);

    \begin{scope}
    \draw[fill = redcol]
    (0, -2.5) ellipse (1.618cm and 1cm);
    \clip (0, -2.5) ellipse (1.618cm and 1cm);
    \fill[fill = bluecol] (9775/3000, -2.5) circle (9775/3000+0.15);
    \end{scope}
    \draw (0, -2.5) ellipse (1.618cm and 1cm) node[normal] {$N(v_{i+1})\cap V_{i+1}$};
    \draw[fill = greencol] (2.5, 1.25) ellipse (1cm and 0.618cm) node[normal] {$C$};
    \end{tikzpicture}
    \caption{The situation of Lemma~\ref{lem:triangle-free-candidates-characterization}. 
    The set $N[D^*]\cap V_{i+1}$ is depicted in gray (except $v_{i+1}$, in the first case of the proof), and the set $S = V_{i+1}\setminus (N[D^*] \cup \{v_{i+1}\})\subseteq N(v_{i+1})\cap V_{i+1}$ is in white. Note that $C=N(S)\setminus \{v_{i+1}\}$ may intersect $V(G)\setminus V_{i+1}$.}
    \label{fig:characterization}
\end{figure}
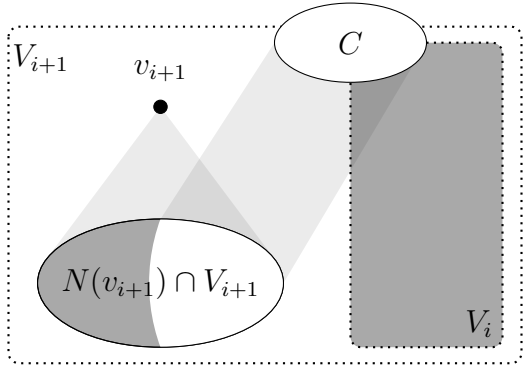

We now show how to efficiently enumerate the candidate extensions.

\begin{lemma}\label{lem:triangle-free-candidates-enumeration}
    There is an algorithm enumerating $\C(D^*,i)$ in total time 
    $O(n^4 \cdot|\D(G,A)|)$ and $O(n^2)$ space whenever $N(v_{i+1})$ is an independent set.
\end{lemma}
\begin{proof}
First, observe that given any set $B$ of vertices, we can test in $O(n^2)$ time and space whether $B\in \C(D^*,i)$.
Hence, it suffices to enumerate in total time $O(n^4 \cdot|\D(G,A)|)$ and $O(n^2)$ space a superset $\F$ of $\C(D^*,i)$ of size $O(n^2 \cdot |\D(G,A)|)$, and for each element of $\F$ to test whether it belongs to $\C(D^*,i)$. By Lemma~\ref{lem:triangle-free-candidates-characterization} we can use
$$\F=\{v_{i+1}\}\cup \D(H)\cup \{D\cup \{u\}\mid D\in \D(H),\ u\in V(G)\},$$
where $H$ is the split graph defined in the statement of Lemma~\ref{lem:triangle-free-candidates-characterization}.
Observe that $$|\F|\leq (n+1) \cdot |\D(H)|+1$$ and,
using Proposition~\ref{prop:split-properties}, we can enumerate $\F$ in total time $O(n^3\cdot |\D(H)|)$ and space $O(n^2)$.
It now remains to observe that by Lemma~\ref{lem:triangle-free-candidates-characterization} and Corollary~\ref{cor:cibound}
we have $$|\D(H)|\leq n\cdot |\C(D^*,i)|+1\leq n\cdot |\D(G,A)|+1,$$ so the claimed time complexity follows.
\end{proof}

We conclude with the following theorem that we state in a more general way than in Section~\ref{sec:intro}, and which is a consequence of Theorem~\ref{thm:ordered-generation}, Lemma~\ref{lem:triangle-free-candidates-enumeration}, and of the fact that when $G[A]$ is triangle-free, the neighborhood of any vertex is an independent set.

\begin{theorem}\label{thm:triangle-free}
    There is an algorithm that, given a bicolored graph $G(A)$ on $n$ vertices such that $G[A]$ is triangle-free, enumerates the set $\D(G,A)$ in time
    \[
        O(\poly(n) \cdot |\D(G,A)|^2)
    \]
    and $O(n^3)$ space.
\end{theorem}

When $A = V(G)$, we have $\D(G) = \D(G,A)$. Hence,
Theorem~\ref{thm:triangle-free} implies the existence of an algorithm
enumerating the minimal dominating sets in
triangle-free graphs in output-polynomial time and polynomial space.

\section{Minimal dominating sets in \texorpdfstring{$K_t$}{Kt}-free graphs}
\label{sec:kt-free}

In this section, we generalize the characterization of Lemma~\ref{lem:triangle-free-candidates-characterization} and show how to use it to enumerate minimal dominating sets in~$K_t$-free graphs, at the cost of an increased complexity (see Theorem~\ref{thm:kt-free}).

We start with a general lemma that, roughly, implies that any output-polynomial time algorithm that may repeat outputs can be turned into one without repetition, without increasing space.

\begin{lemma}\label{lem:avoid-repetitions}
    Let $\setin{}, \setout{}$ be two sets and $R\subseteq \setin{} \times \setout{}$ be a relation.
    Let $f,s \colon \setin{} \to \N$ be two functions.
    Suppose that there is a deterministic algorithm enumerating, given any $x \in \setin{}$, the set $\{y \in \setout{} \mid xRy\}$ in time at most $f(x)$ and space at most $s(x)$, possibly with repetition.
    Then there is an algorithm that, on the same input, returns the same output without repetition, in time $O(f(x)^2)$ and space $O(s(n))$.
\end{lemma}

\begin{proof}
Let \algobp{} be the algorithm that on input $x\in \setin{}$ outputs $\{y \in \setout{} \mid xRy\}$, possibly with repetition, in time at most $f(x)$ and space at most $s(x)$. Elements $y\in \setout{}$ satisfying $xRy$ will be called {\em{solutions}}.
We now give an algorithm \algob{} that, on the same input $x$, outputs all solutions without repetition. Algorithm \algob{} simulates \algobp{} while counting its number of output calls. Every time \algobp{} outputs a solution $y$, \algob{} runs a new simulation of \algobp{} to verify whether $y$ was not output by \algobp{} before. This new simulation is terminated at the first attempt of outputting $y$, and for the verification, \algob{} simply checks the output solution counts in both simulations against each other. If $y$ is indeed output by \algobp{} for the first time, then \algob{} also outputs $y$, and otherwise \algobp{} ignores this output and proceeds with the simulation. Thus, \algob{} outputs every solution exactly once: at the first moment when \algobp{} outputs it. The time complexity of \algob{} is $O(f(x)^2)$, because for every step of \algobp{} we run a second simulation of \algobp{} that takes time at most $f(x)$. The space complexity of \algob{} is at most $2\cdot s(x)+O(1)=O(s(x))$, because we need to store the internal data of two simulations of \algobp{} at any time.
\end{proof}

By combining Lemma~\ref{lem:avoid-repetitions} and Theorem~\ref{thm:ordered-generation}, we get the following corollary.

\begin{corollary}\label{cor:candexrep-to-dom}
Let $f \colon \N^2 \to N$ and $s\colon \N \to \N$ be two functions.
Suppose that there is an algorithm that, given a bicolored graph $G(A)$ on $n$ vertices, a peeling $(V_0, \dots, V_{p})$ of $G(A)$, $i \in \intv{0}{p-1}$ and $D^* \in \D(G, V_i)$, enumerates the set $\C(D^*,i)$ in time at most $f(n, |\D(G,A)|)$ and space at most $s(n)$, possibly with repetition.
Then there is an algorithm that, given a bicolored graph $G(A)$ on $n$ vertices, enumerates the set $\D(G,A)$ in time
\[
        O(n^4d^2 + f(n, d)^2\cdot n d)
\]
and space $O(n\cdot s(n))$, where $d = |\D(G,A)|$.
\end{corollary}

The aforementioned generalization of Lemma~\ref{lem:triangle-free-candidates-characterization} is the following.

\begin{lemma}\label{lem:candidates-characterization}
    Let $G(A)$ be a bicolored graph and $(V_0,\allowbreak{}\dots,V_{p})$ be a fixed peeling of $G(A)$ with vertex sequence $(v_1,\dots,v_p)$. Let $i\in \intv{0}{p-1}$, $D^*\in \D(G,V_i)$ and $S=V_{i+1}\setminus (N[D^*]\cup \{v_{i+1}\})$. Then:
    \begin{itemize}
        \item if $v_{i+1}\in N[D^*]$ then $\C(D^*,i)=\D(G,S)$;
        \item otherwise, every element of $\C(D^*,i)$ is of the form $Q\cup \{w\}$ for some $w\in N[v_{i+1}]$ and $Q\in \D(G,S\setminus N[w])$. Furthermore, in this case $|\D(G,S\setminus N[w])|\leq |\C(D^*,i)|$ for each $w\in N[v_{i+1}]$.
    \end{itemize}
\end{lemma}
\begin{proof}
By definition, $\C(D^*,i)=\D(G,V_{i+1}\setminus N[D^*])$. Note that if $v_{i+1}\in N[D^*]$ then $V_{i+1}\setminus N[D^*]=S$, so we immediately get $\C(D^*,i)=\D(G,S)$.
This resolves the first case.

Suppose then that $v_{i+1}\notin N[D^*]$ and consider any $X\in \C(D^*,i)$. Since $X$ minimally dominates $V_{i+1}\setminus N[D^*]=S\cup \{v_{i+1}\}$, there exists some $w\in N[v_{i+1}]\cap X$. Then $X\setminus \{w\}$ dominates $S\setminus N[w]$ and for every element of $X\setminus \{w\}$, its private neighbor in $S\cup \{v_{i+1}\}$ has to actually belong to $S\setminus N[w]$. We conclude that $X\setminus \{w\}\in \D(G,S\setminus N[w])$, proving the characterization of the elements of $\C(D^*,i)$ in this case.

We are left with proving the claimed upper bound on $|\D(G,S\setminus N[w])|$, for each $w\in N[v_{i+1}]$. Take any $Q\in \D(G,S\setminus N[w])$; clearly $w\notin Q$. If $Q$ dominates $S\cup \{v_{i+1}\}$, then $Q$ is also a minimal dominating set of $S\cup \{v_{i+1}\}$, because every vertex of $Q$ has a private neighbor in $S\setminus N[w]\subseteq S\cup \{v_{i+1}\}$. Otherwise, $Q\cup \{w\}$ is a minimal dominating set of $S\cup \{v_{i+1}\}$: $v_{i+1}$ is the private neighbor of $w$, and $w$ could not steal any private neighbors in $S\setminus N[w]$ from any vertices from $Q$. We conclude that either $Q$ or $Q\cup \{w\}$ belongs to $\C(D^*,i)$, which proves that $|\D(G,S\setminus N[w])|\leq |\C(D^*,i)|$.
\end{proof}

Let us point out the key difference between the statements of Lemma~\ref{lem:candidates-characterization} and of Lemma~\ref{lem:triangle-free-candidates-characterization}.
In Lemma~\ref{lem:triangle-free-candidates-characterization}, we reduced the enumeration of $\C(D^*,i)$ to the enumeration of $\D(H)$ for a single split graph $H$. In Lemma~\ref{lem:candidates-characterization}, to obtain larger generality we need to separately consider sets $\D(G,S\setminus N[w])$ for each $w\in N[v_{i+1}]$. When enumerating $\C(D^*,i)$ via enumerating these sets, we will unavoidably obtain repetitions of elements of $\C(D^*,i)$.
These will be handled using Lemma~\ref{lem:avoid-repetitions} at the cost of an increased complexity.

Observe that by Lemma~\ref{lem:candidates-characterization}, to be able to enumerate the candidate extensions in general (and thus the minimal dominating sets, using Theorem~\ref{thm:ordered-generation}) in output-polynomial time, it suffices to be able to enumerate the minimal dominating sets of~$G(S)$ and of $G(S \setminus N[w])$ for every $w\in N[v_{i+1}]$.
For bicolored graphs $G(A)$ such that $G[A]$ is $K_t$-free, this can be done by exploiting the fact that $G[S]$ is $K_{t-1}$-free and running the same algorithm on $G(S)$, as we shall describe now.
We recall that $\omega(G)$ denotes the size of a largest clique in $G$.

\begin{theorem}
\label{thm:kt-free}
    There is a function $p \colon \N \to \N$ and an algorithm that,
    given a bicolored graph $G(A)$ on $n$ vertices, enumerates the set $\D(G,A)$ in time at most
    \[
        p(t) \cdot n^{2^{t+1}} \cdot |\D(G,A)|^{2^t}
    \]
    and space at most $p(t) \cdot n^{t+1}$, where $t=\omega(G[A])+1$.
\end{theorem}

When $A = V(G)$, we have $\D(G) = \D(G,A)$. Hence,
Theorem~\ref{thm:kt-free} implies the existence of an algorithm
enumerating, for every integer $t\geq 1$, the minimal dominating sets in
$K_t$-free graphs in output-polynomial time and polynomial space.
We stress that we provide a single
algorithm for all values of $t$---note that as stated, the algorithm does not require knowledge of $t$.

\begin{proof}[Proof of Theorem~\ref{thm:kt-free}]
In this proof we consider two algorithms \algoa{} and \algob{} that
recursively call each other in order to enumerate the minimal
dominating sets of a bicolored graph. We first give their
specifications, then describe them, and finally prove that they
perform as specified. Let $f \colon \N^3\to \N$ be defined by
$f(n,d,t) = n^{2^{t+1} - 3} \cdot d^{2^t-1}$, for every $n,d,t \in \N$.

\paragraph{Specifications of \algoa{} and \algob{}} 
We will show that
the algorithms $\algoa{}$ and $\algob{}$ have the following properties $P$ and $Q$, for every $t\geq 1$ in case of $P(t)$ and every $t\geq 2$ in case of $Q(t)$.
\begin{itemize}
    \item[$P(t)$:] There is a constant $p(t) \in \N$ such that given
      an $n$-vertex graph $G$ and a set $A \subseteq V(G)$ such that
      $G[A]$ is $K_t$-free, \algoa{} outputs $\D(G,A)$ in time at most
      $p(t) \cdot f(n,|\D(G,A)|,t)$ and space at most $p(t) \cdot n^{t+1}$.
    \item[$Q(t)$:] There is a constant $q(t)\in \N$ such that given a
      $n$-vertex graph $G$, a set $A \subseteq V(G)$ such that $G[A]$
      is $K_t$-free, a peeling $(V_0, \dots, V_p)$ of $G(A)$
      with vertex sequence $(v_1, \dots, v_p)$, $i \in \intv{0}{p-1}$,
      and $D^* \in \D(G, V_i)$, \algob{} outputs $\C(D^*,i)$ in time
      at most $q(t) \cdot n^2 \cdot f(n, |\D(G,A)|, t-1)^2$ and space at most $q(t)
      \cdot n^{t}$.
    \end{itemize}

The statement of Theorem~\ref{thm:kt-free} is
implied by $P(t)$ holding for all $t\geq 1$. In order to prove it, we
will also show that $Q(t)$ holds for every $t\geq 2$.
Let us first describe~\algoa{}.
\paragraph{Description of \algoa{}}
The algorithm \algoa{} is the one given by Theorem~\ref{thm:ordered-generation} that takes as input a bicolored graph $G(A)$, using \algob{} as a routine to enumerate candidate extensions. We will show below that \algob{} indeed does so.

\paragraph{Description of \algob{}}
Recall that \algob{} takes as input a bicolored graph $G(A)$, a peeling $(V_0, \dots, V_p)$ of $G(A)$ with vertex sequence $(v_1, \dots, v_p)$, an integer $i \in \intv{0}{p-1}$, and a set $D^* \in \D(G, V_i)$.
  
We first describe an auxiliary routine~\algobp{}.
Let $S = V_{i+1}\setminus (N[D^*]\cup \{v_{i+1}\})$.
Lemma~\ref{lem:candidates-characterization} above allows us to consider two cases depending on whether $D^*$ dominates $v_{i+1}$ or not:
\begin{enumerate}[(i)]
\item \label{it:vidom} if $v_{i+1}\in N[D^*]$, we call algorithm \algoa{} on $G(S)$ to enumerate $\D(G,S)$ and we give the same output;
\item \label{it:repet} otherwise, we iterate over all $w \in N[v_{i+1}]$ and $Q \in \D(G, S\setminus N[w])$ (where the latter is obtained via a call to \algoa{}) and output $Q\cup \{w\}$ if and only if it is a candidate extension of $D^*$.
\end{enumerate}
We are now done with~\algobp{}.
As we will show later, \algobp{} enumerates $\C(D^*, i)$, however each element may be repeated, up to $n$ times.
Then \algob{} is obtained from \algobp{} using Lemma~\ref{lem:avoid-repetitions}. This concludes the description of~\algob{}.

\paragraph{Correctness of \algoa{} and \algob{}} 
Now that we described the algorithms \algoa{} and \algob{}, we show that they
conform to their specifications, \ie, we prove that $P(t)$ holds for
every $t\geq 1$ and that $Q(t)$ holds for every $t \geq 2$.
The proof by induction on $t$ is split in lemmas.

\begin{lemma}\label{lem:base}
  $P(1)$ holds.
\end{lemma}

\begin{proof}
  The statement $P(1)$ deals with pairs $(G,A)$ such that $G[A]$ is
  $K_1$-free, so $A = \emptyset$. In these cases we clearly have
  $\D(G,A) = \{\emptyset\}$. Notice that such inputs are correctly handled by algorithm \algoa{}. 
Checking whether $A$ is empty and returning $\{\emptyset\}$ takes $O(n)$ time and $O(n^2)$ space. We define $p(1)$
as an integer such that these steps take at most $p(1) \cdot n$ time
and at most $p(1) \cdot n^2$ space
on an input graph of order~$n$. As $f(n,|\D(G,A)|,t) = n$ in this case, $P(1)$ holds.
\end{proof}

\begin{lemma}\label{lem:ind1}
  For every integer $t \geq 1$, $P(t) \Rightarrow Q(t +1)$.
\end{lemma}
\begin{proof}
  Let $t\geq 1$ and let us assume that the statement $P(t)$
  holds (in particular, $p(t)$ is defined).
  Let $\mathcal{I} = (G, A, V_0, \dots,\allowbreak{} V_{p}, v_1, \dots, v_p, i, D^*)$ be an
  input of \algob{} such that $G[A]$ is $K_{t+1}$-free.
  Let us define $n = |G|$ and $d = |\D(G,A)|$.
  We review the description of \algob{} to show that $Q(t+1)$
  holds. We first consider the auxiliary routine \algobp{}.
  
  \begin{claim}\label{claim:auxi}
    Given $\mathcal{I}$, the algorithm \algobp{} enumerates $\C(D^*, i)$ with each output repeated up to $n$~times, in time at most $k \cdot n
    \cdot f(n, d, t)$ and space at most $k \cdot n^{t+1}$, for some constant~$k$.
  \end{claim}
  \begin{proof}
  Let $S = V_{i+1}\setminus(\{v_{i+1}\} \cup N[D^*])$.
  Note that as $D^*$ dominates $V_i$, we have $S \subseteq N(v_{i+1})\cap V_{i+1}$. Also, $S$ can be computed in $O(n^2)$ time and space.
  
  Since $i<p$, we have $V_{i+1} \subseteq A$, from the definition of
  a peeling. In
  particular, $G[V_{i+1}]$ is $K_{t+1}$-free. As $S \subseteq
  N(v_{i+1})\cap V_{i+1}$, we get that $G[S]$ is $K_t$-free.
  Hence, when applying the algorithm recursively to enumerate $\D(G,S')$ for any $S'\subseteq S$, we may use the already established property $P(t)$, yielding the following:
  
  \begin{remark}\label{rem:ih}
    For any $S' \subseteq S$, a call to \algoa{}
  on $(G,S')$ returns $\D(G,S')$ in time at most $p(t) \cdot f(n,|\D(G,S')|,t)$
  and space at most $p(t) \cdot n^{t+1}$.
  \end{remark}
  
  If $v_{i+1}\in N[D^*]$, then, by Lemma~\ref{lem:candidates-characterization}, we enumerate $\C(D^*,i)$ without repetitions simply by enumerating $\D(G,S)$. By Remark~\ref{rem:ih}, this takes time
  \begin{align*}
    & p(t) \cdot f(n, |\D(G,S)|,t)\\
    =~& p(t) \cdot f(n, |\C(D^*, i)|,t)& \text{(by Lemma~\ref{lem:candidates-characterization})}\\
    \leq~ & p(t) \cdot f(n, d, t) & \text{(by Corollary~\ref{cor:cibound})}
    \end{align*}
    and space at most $p(t) \cdot n^{t+1}$.

 Now suppose that $v_{i+1}\notin N[D^*]$. Then, by Lemma~\ref{lem:candidates-characterization}, we enumerate all elements of the set $\C(D^*,i)$, however each of them is enumerated once per every form $Q\cup \{w\}$ it can take, where $w\in N[v_{i+1}]$ and $Q\in \D(G,S\setminus N[w])$. Every such occurrence is characterized by the choice of $w$, hence there are at most $n$ of them and, consequently, every member of $\C(D^*,i)$ is enumerated at most $n$ times.
 
 Regarding time and space
  complexity we perform at most $n$ times (once for every
  choice of $w\in N[v_{i+1}]$) the following operations:
  \begin{itemize}
  \item constructing $S \setminus N[w]$, in $O(n)$ time and
    space;
  \item calling \algoa{} on $(G, S \setminus N[w])$, in time at most
 $p(t) \cdot f(n, |\D(G,S\setminus N[w])|, t)$ and space at most $p(t) \cdot n^{t+1}$, by
    Remark~\ref{rem:ih};
  \item checking, for each set $Q$ among the outputs of \algoa{}, whether $Q\cup \{w\}$ belongs to $\C(D^*,i)$, in $O(n^2)$ time and space. There are at most $d$ outputs.
  \end{itemize}
  By Lemma~\ref{lem:candidates-characterization} and Corollary~\ref{cor:cibound}, we have
  \begin{align*}
      |\D(G,S\setminus N[w])|\leq |\C(D^*,i)|\leq d.
  \end{align*}
  Therefore, in total the time complexity of these steps adds up to:
  \begin{align}
   & n \cdot \left [ O(n) + p(t) \cdot f(n, d, t) + O\left (n^2 \cdot d
     \right ) \right ]\nonumber\\
    =~& O\left (p(t) \cdot n \cdot f(n, d, t) \right ) & \text{(as $t \geq 2$ in this case).} \label{eq:timebp}
  \end{align}
  Similarly, the space complexity can be upper-bounded by $O(p(t) \cdot
  n^{t+1})$. We conclude by setting $k$ as $p(t)$ times the maximum of the constants hidden in the $O(\cdot)$ notation above.
  \cqed
  \end{proof}
 As proved in Lemma~\ref{lem:avoid-repetitions}, the algorithm of Claim~\ref{claim:auxi} can be turned into an algorithm \algob{} that does not repeat outputs. That is, there is a constant $q(t+1)$ (depending on $k$) such that given $\mathcal{I}$, \algob{} runs in time at most $q(t+1) \cdot n^2\cdot f(n, d, t)^2$ and space at most $q(t+1) \cdot n^{t+1}$. Hence $Q(t+1)$ holds, as desired.
\end{proof}

\begin{lemma}\label{lem:ind2}
  For every integer $t \geq 2$, $Q(t) \Rightarrow P(t)$.
\end{lemma}

\begin{proof}
   Let us assume that for some integer $t\geq 2$, the statement $Q(t)$
   holds (and in particular $q(t)$ is defined).
   Let $G$ be a graph and $A \subseteq V(G)$ be such that $G[A]$ is $K_t$-free.
   We set $n = |G|$ and $d = |\D(G,A)|$.
   By $Q(t)$, the enumeration of candidate extensions in $G(A)$ can be carried out by \algob{} in total time at most \[q(t)\cdot n^2 \cdot f(n, d, t-1)^2\] and space at most $q(t) \cdot n^{t}$.
   According to Theorem~\ref{thm:ordered-generation}, \algoa{} then enumerates $\D(G,A)$ in time
    \begin{align*}
       & O(n^4\cdot d^2 + q(t)\cdot n^3 \cdot f(n, d, t-1)^2 \cdot d)\\
       =~& O(n^4\cdot d^2 + q(t) \cdot f(n,d,t)) & \text{(by the definition of $f$)}\\
       =~& O(q(t) \cdot f(n,d,t)) & \text{(as $t\geq 2$)}
   \end{align*}
   and space $O(q(t) \cdot n^{t+1})$. Therefore, there is a constant $p(t)$ (depending on $q(t)$) such that \algoa{} runs on this input in time at most 
$p(t) \cdot f(n,d,t)$  and space at most $p(t) \cdot n^{t+1}$. This proves $P(t)$.
\end{proof}
  
\paragraph{Concluding the proof.}
We proceed by induction on $t$. The base case $P(1)$ follows from Lemma~\ref{lem:base}. The induction step that, for every integer $t \geq 1$, $P(t)$ implies $P(t+1)$, is obtained by combining Lemmas~\ref{lem:ind1} and \ref{lem:ind2}. We conclude that $P(t)$ holds for every integer $t\geq 1$. That is, the algorithm \algoa{} has the properties claimed in the statement of the theorem.
\end{proof}

We note that the complexity of the algorithm of Theorem~\ref{thm:kt-free} for $K_t$-free graphs could be slightly improved when $t \geq 3$, using Theorem~\ref{thm:triangle-free} as a base case, however that would not remove the exponential contribution of $t$ to the degree of the polynomial.

\section{Variants of \texorpdfstring{$K_t$}{Kt}-free graphs}\label{sec:variants}

We give output-polynomial time algorithms for classes related to $K_t$-free graphs relying on the algorithms and characterizations of candidate extensions given in Sections~\ref{sec:bt}, \ref{sec:triangle-free}, and~\ref{sec:kt-free}.

\subsection{Forbidding \texorpdfstring{$K_t+K_2$}{Kt+e}}

In this section we show how the algorithm of Theorem~\ref{thm:kt-free} on $K_t$-free graphs can be extended to the setting of $(K_t+K_2)$-free graphs.

\begin{theorem}\label{thm:ktme}
    There is an algorithm that, for every fixed $t \in \N$, enumerates minimal dominating sets in $(K_t+K_2)$-free graphs in output-polynomial time and polynomial space.
\end{theorem}

\begin{proof}
Let $t \in \N$ and let $G$ be a $(K_t+K_2)$-free graph.
It is well-known that the minimal dominating sets of $G$ that induce edgeless subgraphs are exactly the maximal independent sets of~$G$.
We can therefore enumerate these using the polynomial delay algorithms of Tsukiyama et al.~\cite{tsukiyama1977new} for maximal independent sets. In the sequel we may thus focus on those minimal dominating sets of $G$ that induce at least one edge.

We show how to enumerate, for every edge $uv$ of $G$, the minimal dominating sets of $G$ that contain both $u$ and $v$.
Let $A_{uv}=V(G) \setminus N[\{u,v\}]$ and observe that $G[A_{uv}]$ is $K_t$-free.
First, we enumerate $G(A_{uv})$ using the algorithm of Theorem~\ref{thm:kt-free}, which runs in output-polynomial time and polynomial space, as $t$ is fixed.
For every $D\in \D(G,A_{uv})$ obtained from the aforementioned call, we output $D\cup\{u,v\}$ if it is a minimal dominating set of $G$, and discard $D$ otherwise.
By Lemma~\ref{lem:cand-ext-bound} (applied for $H=G$ and $B = V(G)$) we have 
$
    |\D(G,A_{uv})|\leq |\D(G)|
$.
Hence, enumerating $\D(G,A_{uv})$ produces all those minimal dominating sets of $G$ that at least induce the edge $uv$ in time $\poly(n \cdot |\D(G)|)$ and space $\poly(n)$, where the degrees of these polynomials depend on $t$ (see Theorem~\ref{thm:kt-free}).

Now that we know how to enumerate minimal dominating sets that induce at least one particular edge, we can run the above routine for every edge of $G$ to enumerate all minimal dominating sets of $G$, possibly with repetitions. 
Observe that the same output can be repeated at most $|E(G)|$ times. 
Then, repetitions are avoided using Lemma~\ref{lem:avoid-repetitions} with $\setin{}$ being the set of all graphs, $\setout{}$ the set of all vertex sets, and $R$ the relation that associates every graph to its minimal dominating sets.
\end{proof}

\subsection{Forbidding \texorpdfstring{$K_t-e$}{Kt-e}}

Another interesting case is the one of $(K_t-e)$-free graphs.
In this section we show how the characterization of Lemma~\ref{lem:candidates-characterization} can be used to enumerate candidate extensions in diamond-free graphs (which are $(K_t-e)$-free for $t=4$), which by Theorem~\ref{thm:ordered-generation} gives an output-polynomial time algorithm enumerating minimal dominating sets in this class.
We leave open the existence of such an algorithm in the case when $t\geq 5$.

In what follows, we consider a bicolored graph $G(A)$ on $n$ vertices such that $G$ is diamond-free, together with a fixed peeling $(V_0,\dots, V_{p})$ of $G(A)$ with vertex sequence $(v_1,\dots,v_p)$.
Then, we consider
\begin{align*}
    & i \in \intv{0}{p-1},
    & D^* \in \D(G,V_i),
\end{align*}
and define $S=V_{i+1}\setminus(N[D^*]\cup \{v_{i+1}\})$ and $\C(D^*,i)$ as in Sections~\ref{sec:bt}, \ref{sec:triangle-free}, and~\ref{sec:kt-free}.
Note that contrarily to the triangle-free case and the $K_t$-free case, we here require the whole graph $G$ to be diamond-free and not only~$G[A]$.
We start with an easy observation.

\begin{observation}\label{obs:clique-partition}
    For every vertex $u$ of $G$, $G[N(u)]$ is $P_3$-free. 
    Then $G[N(v_{i+1})]$, hence also $G[S]$, can be partitioned into a disjoint union of cliques (\ie, it is a cluster graph).
\end{observation}

We will show how to minimally dominate one clique of $S$, then a disjoint union of cliques of $S$, and will conclude with the enumeration of $\C(D^*,i)$.

\begin{lemma}\label{lem:clique-neighborhoods}
    Let $K$ be a clique of $G[S]$ and $u$ be a vertex in $G-S$, $u\neq v_{i+1}$, that is adjacent to some vertex of~$K$.
    If $u$ is adjacent to $v_{i+1}$, then it is complete to $K$.
    Otherwise $u$ has exactly one neighbor in $K$.
\end{lemma}
\begin{proof}
    If $u\in N(v_{i+1})$ then, as $G[N(v_{i+1})]$ is $P_3$-free and $K\subseteq N(v_{i+1})$, $u$ is complete to~$K$.
    If $u$ is not adjacent to $v_{i+1}$, then it has exactly one neighbor in $K$, as otherwise  $\{a,b,u,v_{i+1}\}$ would induce a diamond in $G$, for any two neighbors $a,b\in K$ of~$u$.
\end{proof}

\begin{lemma}\label{lem:clique-domsetsenumeration}
    Let $K$ be a clique in $G[S]$. 
    Then $\D(G,K)$ can be enumerated in total time 
    $O(n^2+ n\cdot |\D(G,K)|)$ and $O(n^2)$ space.
\end{lemma}
\begin{proof}
    We describe an algorithm enumerating $\D(G,K)$ in the specified time and space bounds.
    We first output $\{v_{i+1}\}$ as it is complete to~$K$.
    We then output all vertices $u\in N(v_{i+1})$ such that $u\in K$ or $u$ is adjacent to some vertex of~$K$.
    By Lemma~\ref{lem:clique-neighborhoods}, these vertices are also complete to~$K$.
    Then, for every $x\in K$, we compute the neighborhood of $x$ outside of $N(v_{i+1})$ in total time $O(n^2)$.
    By Lemma~\ref{lem:clique-neighborhoods}, these neighborhoods are disjoint.
    At last, we enumerate the unordered Cartesian products of 
    these neighborhoods. This can clearly be done in total time of $O(n\cdot |\D(G,K)|)$ using $O(n)$ space as they are disjoint.
    Clearly, every element in such an unordered Cartesian product is a minimal dominating set of $K$, and the described algorithm performs within the specified time and space bounds.
    The correctness of the algorithm follows from Lemma~\ref{lem:clique-neighborhoods}.
\end{proof}

\begin{lemma}\label{lem:P3-free}
    Let $W$ be a subset of $S$.
    Then $\D(G,W)$ can be enumerated in total time
    $O(n^{7}\cdot |\D(G,A)|^3)$ and $O(n^3)$ space.
\end{lemma}
\begin{proof}
    We use the ordered generation described in Section~\ref{sec:bt}.
    The algorithm first computes a peeling $(U_1,\dots,U_{q})$ of $G(W)$ with vertex sequence $(u_1,\dots,u_q)$, in $O(n^2)$ time and space.
    Note that $N[u_1]\cap W,\dots,N[u_q]\cap W$ is exactly the disjoint clique partition of $G[W]$; denote these sets by $W_1,\ldots,W_q$.
    Given $j\in \intv{0}{q-1}$ and $D^\circ\in \D(G,U_j)$, we define $\Cp(D^\circ, j)$ as the set of candidate extensions of $(D^\circ, j)$ with respect to the chosen peeling of $G(W)$
    and we show how to enumerate $\Cp(D^\circ,j)$ in time
    $
        O(n^6\cdot |\D(G,A)|^2)
    $
    and using $O(n^2)$ space.
    
    We rely on the same characterization of candidate extensions that we use in the proof of Theorem~\ref{thm:kt-free}, \ie{}, Lemma~\ref{lem:candidates-characterization}. Recall that this lemma allows us to consider two cases depending on whether $D^\circ$ dominates $u_{j+1}$ or not. Let $Y=W_{j+1}\setminus N[D^\circ]$.
    
    If $D^\circ$ dominates $u_{j+1}$, then we can enumerate $\C'(D^\circ,j)$ by just enumerating $\D(G,Y)$, as these sets coincide. As $Y$ is a clique in $G[S]$, by Lemma~\ref{lem:clique-domsetsenumeration} we can enumerate $\D(G,Y)$ in time $O(n^2+n\cdot |\D(G,Y)|)$ and space $O(n^2)$. 
    By Lemma~\ref{lem:cand-ext-bound} we have $|\D(G,Y)|\leq |\D(G,A)|$, hence the procedure runs within the required time and space complexity.
    
    In the remaining case when $D^\circ$ does not dominate $u_{j+1}$, we iterate over all $w \in N[u_{j+1}]$ and $Q \in \D(G, Y\setminus N[w])$ (obtained via a call to the algorithm of Lemma~\ref{lem:clique-domsetsenumeration}) and output $Q\cup \{w\}$ if and only if this set belongs to $\C'(D^\circ,j)$, which can be checked in $O(n^2)$ time and space.
    Again, by Lemma~\ref{lem:cand-ext-bound} we have $|\D(G,Y\setminus N[w])|\leq |\D(G,A)|$ for all $w$ as above. Hence, 
    in total, the described algorithm enumerates $\Cp(D^\circ,i)$, possibly with repetitions, in time $O(n^3\cdot |\D(G,A)|)$ and using $O(n^2)$ space.
    Using Corollary~\ref{cor:candexrep-to-dom}, we obtain an algorithm enumerating $\D(G,W)$ in time
    $O(n^{7}\cdot |\D(G,A)|^3)$, and using $O(n^3)$ space.
\end{proof}

\begin{lemma}\label{lem:diamond-free-candidates-enumeration}
    There is an algorithm enumerating $\C(D^*,i)$, possibly with repetitions, in total time 
    $O(n^{8}\cdot |\D(G,A)|^3)$ and using $O(n^3)$ space.
\end{lemma}
\begin{proof}
    We apply the same argument as in the previous lemma, and in the proof of Theorem~\ref{thm:kt-free}.
    Lemma~\ref{lem:candidates-characterization} allows us to consider two cases depending on whether $D^*$ dominates $v_{i+1}$ or not.
    If $D^*$ dominates $v_{i+1}$, we call the algorithm of Lemma~\ref{lem:P3-free} to enumerate $\D(G,S)$ without repetitions in total time $O(n^{7}\cdot |\D(G,A)|^3)$
    and $O(n^3)$ space.
    If $D^*$ does not dominate $v_{i+1}$, we iterate over all $w \in N[v_{i+1}]$ and $Q \in \D(G, S\setminus N[w])$ (obtained via a call to the algorithm of Lemma~\ref{lem:P3-free} as $S\setminus N[w]\subseteq S$) and output $Q\cup \{w\}$ if and only if this set belongs to $\C(D^*,i)$.
    An analogous complexity analysis shows that this algorithm runs in time $O(n^{8}\cdot |\D(G,A)|^3)$ and uses $O(n^3)$ space, and it enumerates $\C(D^*,i)$ possibly with repetitions.
\end{proof}

As a consequence of Corollary~\ref{cor:candexrep-to-dom} and Lemma~\ref{lem:diamond-free-candidates-enumeration}, we get the following. 

\begin{theorem}\label{thm:diamond-free}
    There is an algorithm that,
    given a bicolored graph $G(A)$ on $n$ vertices such that $G$ is diamond-free,
    enumerates the set $\D(G,A)$ in time
    \[
        O(\poly(n) \cdot |\D(G,A)|^7)
    \]
    and $O(n^4)$ space.
\end{theorem}

Note that when $A = V(G)$, we have $\D(G)=\D(G,A)$. 
Hence, Theorem~\ref{thm:diamond-free} implies the existence of an algorithm enumerating the minimal dominating sets in
diamond-free graphs in output-polynomial time and using polynomial space, which is one of the two cases of Theorem~\ref{th:pawdiam}.

\subsection{Paw-free graphs}

We now consider the exclusion of a specific graph, the paw, and show that \DomEnum{} admits an output-polynomial time algorithm in paw-free graphs.

In what follows, we consider a bicolored graph $G(A)$ on $n$ vertices such that $G$ is paw-free, together with a fixed peeling $(V_0,\dots, V_{p})$ of $G(A)$ with vertex sequence $(v_1,\dots,v_p)$.
Then, we consider
\begin{align*}
    & i \in \intv{0}{p-1},
    & D^* \in \D(G,V_i),
\end{align*}
and define $S=V_{i+1}\setminus(N[D^*]\cup \{v_{i+1}\})$ and $\C(D^*,i)$ as in Sections~\ref{sec:bt}, \ref{sec:triangle-free} and~\ref{sec:kt-free}.
As in the previous section we stress that we require the whole graph $G$ to be paw-free, and not only~$G[A]$.
We start with an easy observation.

\begin{observation}
    For every vertex $u$ of $G$, $G[N(u)]$ is $\overline{P_3}$-free.
    Hence $G[S]$ is a complete multipartite graph (also called a co-cluster graph).
\end{observation}

Note that when enumerating $\C(D^*,i)$ (\ie{}, the minimal dominating sets of $G(S)$), we may safely ignore the edges between two vertices of $G-S$.
Therefore, if $S$ is an independent set, we can delete all edges of $G-S$ (in $O(n^2)$ time) to obtain that $N(v_{i+1})$ is an independent set and then apply the algorithm enumerating $\C(D^*,i)$ in this setting given by Lemma~\ref{lem:triangle-free-candidates-enumeration}.
In the next lemma, we consider the case where $S$ contains at least one edge. 

\begin{lemma}\label{lem:paw-free-edge}
    Assume that $G[S]$ contains at least one edge, and let $u$ be a vertex of $G$, $u\neq v_{i+1}$, that has a neighbor in~$S$.
    If $u$ is not adjacent to $v_{i+1}$, then it is complete to~$S$.
    Otherwise, $u$ is complete to $S\setminus I_j$, for some $j\in \intv{1}{q}$ where $I_1,\dots,I_q$ is the complete multipartition of $G[S]$ (every $I_j$ induces an independent set in $G$, while vertices from different $I_j$'s are adjacent).
\end{lemma}

\begin{proof}
    As by assumption $G[S]$ contains an edge, we have $q\geq 2$.
    Assume first that $u$ is not adjacent to $v_{i+1}$, but has a neighbor in $S$; in particular $u\notin S$. Suppose for contradiction that $u$ is not complete to $S$. Hence there are vertices $x \in S \cap N(u)$ and $ y\in S \setminus N(u)$.
    Note that $xy\not\in E(G)$ as otherwise $\{u, v_{i+1}, x, y\}$ induces a paw in $G$.
    Then $x,y\in I_j$ for some $j\in \intv{1}{q}$.
    Let $z \in S\setminus I_j$; such a vertex exists as $q \geq 2$ and it is complete to $\{v_{i+1}, x, y\}$ by definition of the~$I_k$'s.
    Then, either $uz \in E(G)$ and $\{u, x, y,z\}$ induces a paw, or $uz \notin E(G)$ and $\{u, v_{i+1}, y,z \}$ does, a contradiction.
    
    Assume now that $u$ is adjacent to $v_{i+1}$. 
    If $u$ belongs to $S$, then it belongs to some~$I_j$, $j\in\intv{1}{q}$ and is complete to $S\setminus I_j$, by definition of the $I_k$'s.
    We now assume $u \in N(v_{i+1}) \setminus S$. If there is no $j\in\intv{1}{q}$ such that $u$ is complete to $S\setminus I_j$, then, as $q\geq 2$, $u$ has at least two non-neighbors $x\in I_{j'}$ and $y\in I_{j''}$ for two different $j',j''\in \intv{1}{q}$.
    Then $\{u, v_{i+1}, x, y\}$ induces a paw in $G$, a contradiction.
\end{proof}

\begin{lemma}\label{lem:paw-free-candidates-enumeration}
    There is an algorithm enumerating $\C(D^*,i)$ in total time $O(n^5 \cdot|\D(G,A)|)$ and $O(n^2)$ space.
\end{lemma}

\begin{proof}
In the case where $S$ induces an independent set, we use the algorithm of Lemma~\ref{lem:triangle-free-candidates-enumeration} to enumerate $\C(D^*,i)$ in time 
\[
    O(n^4\cdot|\D(G,A)|)
\]
and $O(n^2)$ space.
Otherwise, we deduce from Lemma~\ref{lem:paw-free-edge} that minimal dominating sets of $S$ are either of size at most two, or of the form $I_j$ for some $j\in \intv{1}{q}$.
If $v_{i+1}\in N[D^*]$, that is if $S=V_{i+1}\setminus N[D^*]$, we try each of these sets and output those that minimally dominate $S$; this can be done in total time $O(n^4)$.
This enumerates $\C(D^*,i)$ by definition.
If $v_{i+1}\not\in N[D^*]$, we first output $I_j$ for every $j\in \intv{1}{q}$.
Then, we iterate over all vertex subsets of size at most three and output those that minimally dominate $S$; this can be done in total time $O(n^5)$.
This will enumerate $\C(D^*,i)$, for the following reason implied by Lemma~\ref{lem:paw-free-edge}. If $X\in \C(D^*,i)$, then either $X=I_j$ for some $j\in \intv{1}{j}$, or $X$ contains at most three vertices: one with $v_{i+1}$ as a private neighbor and at most $2$ with private neighbors in $S$.
\end{proof}

As a consequence of Theorem~\ref{thm:ordered-generation} and Lemma~\ref{lem:paw-free-candidates-enumeration}, we get the following.

\begin{theorem}\label{thm:paw-free}
    There is an algorithm that, given a bicolored graph $G(A)$ on $n$ vertices such that $G$ is paw-free, enumerates the set $\D(G,A)$ in time
    \[
        O(\poly(n) \cdot |\D(G,A)|^2)
    \]
    and $O(n^3)$ space.
\end{theorem}

Note that when $A = V(G)$, we have $\D(G)=\D(G,A)$. 
Hence, Theorem~\ref{thm:paw-free} implies the existence of an algorithm enumerating the minimal dominating sets in
paw-free graphs in output-polynomial time and using polynomial space, which is the second of the two cases of Theorem~\ref{th:pawdiam}.

\section{Technique limitations}\label{sec:beyond}

In this section, we discuss various obstacles that we detected in our attempts to improve our results or proofs.

\subsection{A standard technique fails for bipartite graphs}
A natural technique (sometimes called {\em flashlight search} or {\em backtrack}) to enumerate valid solutions to a given problem such as, for instance,
sets of vertices satisfying a given property, is to build them element
by element. If during the construction one detects that the current
partial solution
cannot be extended into a valid one, then it can be discarded along with all
the other partial solutions that contain it. 
Note that in order to apply this technique, one should be able to decide whether a given partial solution can be completed into a valid one. 
It turns out that for minimal dominating
sets, this problem (called the {\em extension problem} problem) is \NP{}-complete \cite{kante2011split}, even when
restricted to split graphs \cite{kante2015chordal}. 
We show that it remains \NP{}-complete in bipartite graphs,
so in particular on $(K_t+K_2)$-free graphs for every $t\geq 3$. This stands in contrast with Theorem~\ref{thm:op} and suggests that, indeed, a more involved technique was needed to obtain our results.

The extension problem, denoted \Dcs{}, is formally defined as follows. 
Given a graph $G$ and a set $A\subseteq V(G)$ of vertices, is there a minimal dominating set $D\in \D(G)$ such that $A\subseteq D$. 

\begin{theorem}\label{thm:ext}
    \Dcs{} restricted to bipartite graphs is \NP{}-complete.
\end{theorem}

\begin{proof}
Since \Dcs{} is \NP{}-complete in the general case, it is clear that \Dcs{} is in \NP{} even when restricted to bipartite graphs. Let us now present a hardness reduction from \SAT{}.

Given an instance $\varphi$ of \SAT{} with variables $x_1,\dots,x_n$ and clauses $C_1,\dots,C_m$, we construct a bipartite graph $G$ and a set $A\subseteq V(G)$ such that there exists a minimal dominating set containing $A$ if and only if there exists a truth assignment to the variables of $\varphi$ that satisfies all the clauses.
The graph $G$ has vertex bipartition $(X,Y)$, defined as follows.
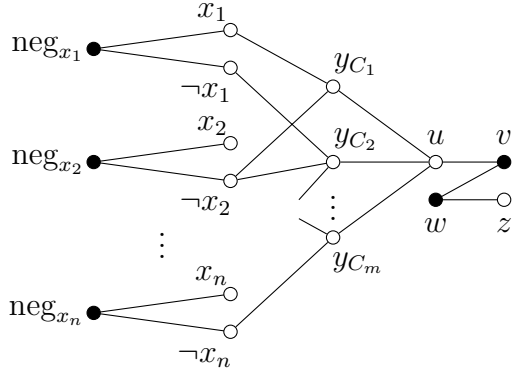
\begin{figure}
    \centering
    \begin{tikzpicture}[xscale = \figscaleCx, yscale = \figscaleCy]
    \draw[every node/.style = black node]
        (0, -0.5) node[label=180:$\mathrm{neg}_{x_1}$] (n1) {}
        (0, -2) node[label=180:$\mathrm{neg}_{x_2}$] (n2) {}
        (0, -4) node[label=180:$\mathrm{neg}_{x_n}$] (nn) {};
    \draw[every node/.style = white node, xshift = 2cm, yshift = 0.25cm]
        (0, -0.5) node[label=135:$x_1$] (x1) {} -- (n1)
        (0, -2) node[label=135:$x_2$] {} -- (n2)
        (0, -4) node[label=135:$x_n$] {} -- (nn);
    \draw[every node/.style = white node, xshift = 2cm, yshift = -0.25cm]
        (0, -0.5) node[label=225:$\neg x_1$] (x1b) {} -- (n1)
        (0, -2) node[label=225:$\neg x_2$] (x2b) {} -- (n2)
        (0, -4) node[label=225:$\neg x_n$] (xnb) {} -- (nn);
    \draw[every node/.style = white node, xshift = 3.5cm]
        (0, -1) node[label=45:$y_{C_1}$] (y1) {}
        (0, -2) node[label=45:$y_{C_2}$] (y2) {}
        (0, -2.5) node[normal] {\vdots}
        (0, -3) node[label=-45:$y_{C_m}$] (ym) {};
    \draw (5, -2) node[white node, label=90:$u$] (u) {} --
        ++(1,0) node[black node, label=90:$v$] {} --
        ++(-1,-0.5) node[black node, label=-90:$w$] {} --
        ++(1,0) node[white node, label=-90:$z$] {};
    \draw (u) -- (y1) (u) -- (y2) (u) -- (ym)
        (xnb) -- (ym) -- ++(-0.5, 0.25)
        (x1) -- (y1) -- (x2b) -- (y2) -- (x1b) (y2) -- ++(-0.5, -0.5);
    \draw (1, -3) node[normal] {\vdots};
    \end{tikzpicture}
    \caption{A bipartite graph $G$ and a set $A\subseteq V(G)$ constructed from an instance of \SAT{} with variables $x_1,\dots,x_n$ and clauses $C_1,\dots,C_m$. Black vertices constitute the set $A$.
    Then $A$ can be extended into a minimal dominating set $D$ of $G$ if and only if there is a truth assignment of the variable satisfying all the clauses.}
    \label{fig:extension-npc}
\end{figure}

The first part $X$ contains two special vertices $u$ and $w$, and for every variable $x_i$, one vertex for each of the literals $x_i$ and $\neg x_i$.
The second part $Y$ contains one vertex $y_{C_j}$ per clause $C_j$, one vertex $\mathrm{neg}_{x_i}$ per variable $x_i$, and two special vertices $v$ and $z$.
For every $i \in \intv{1}{n}$ we make $\mathrm{neg}_{x_i}$ adjacent to the two literals $x_i$ and $\neg x_i$ and for every $j \in \intv{1}{m}$ we make $y_{C_j}$ adjacent to $u$ and to every literal $C_j$ contains.
Finally, we add edges to form the path $uvwz$ and set $A=\{\mathrm{neg}_{x_1},\dots,\mathrm{neg}_{x_n},v,w\}$. Clearly this graph can be constructed in polynomial time from~$\varphi$.
The construction is illustrated in Figure~\ref{fig:extension-npc}.

Let us show that $A$ can be extended into a minimal dominating set of $G$ if and only if $\varphi$ has a truth assignment that satisfies all the clauses. The proof is split into two claims.
A \emph{partial assignment} of $\varphi$ is a truth assignment of a subset of the variables $x_1, \dots, x_n$. Observe that a partial assignment may satisfy all the clauses (\ie, the values of the non-assigned variables do not matter). A partial assignment that satisfies all the clauses is called a \emph{minimal assignment} if none of its proper subsets satisfies all the clauses.

\begin{claim}\label{clm:domsat}
Let  $S \subseteq \{x_1, \neg x_1, \dots, x_n, \neg x_n\}$ be a set containing at most one literal for each variable. Then $S$ minimally dominates $\{y_{C_1}, \dots, y_{C_m}\}$ if and only if its elements form a minimal assignment of $\varphi$.
\end{claim}
\begin{proof}
Suppose $S$ is as above and $S$ minimally dominates $\{y_{C_1},\ldots,y_{C_m}\}$. Consider any $j \in \intv{1}{m}$. Since $y_{C_j} \notin S$, the set $S$ contains a neighbor $x$ of $y_{C_j}$. By construction, $x$ is a literal appearing in~$C_j$. Hence, the literals present in~$S$ form a partial assignment of the variables of $\varphi$ satisfying all its clauses. Moreover, this partial assignment is minimal by the minimality of $S$.
The proof in the other direction is analogous.
\cqed%
\end{proof}

\begin{claim}\label{cl:minlit}
If $D$ is a minimal dominating set of $G$ containing $A$, then $D \setminus A \subseteq \{x_1, \neg x_1,\allowbreak{} \dots,\allowbreak{} x_n, \neg x_n\}$ and $D$ contains at most one literal of each variable.
\end{claim}
\begin{proof}
Notice that $\priv(A,v) = \{u\}$. If $y_{C_j}$ belongs to $D$ for some $j \in \intv{1}{m}$, then $\priv(D,v) = \emptyset$, a contradiction to the minimality of $D$. For similar reasons $u,z\notin D$. Hence $D\cap \{u,z, y_{C_1}, \dots, y_{C_m}\} = \emptyset$.
Besides, for every $i \in \intv{1}{m}$, $D$ contains at most one of $x_i$ and $\neg x_i$, as otherwise $\priv(D, \mathrm{neg}_{x_i})$ would be empty, again contradicting the minimality of~$D$. This proves the claim.
\cqed%
\end{proof}

If $A$ can be extended into a minimal dominating set $D$ of $G$, then by combining the two claims above, we deduce that $\varphi$ has a truth assignment that satisfies all clauses. Conversely, if $\varphi$ has a truth assignment satisfying all the clauses, then it also has a minimal truth assignment satisfying all the clauses, so there is a set~$S$ as in the statement of Claim~\ref{clm:domsat}. In $S \cup A$, every element of~$S$ has a private 
neighbor, as a consequence of the minimality of~$S$ and the fact that no element of $A$ has a neighbor among the clause vertices. Besides, each of $\mathrm{neg}_{x_1},\dots,\mathrm{neg}_{x_n}$ has a private neighbor (because $S$ contains at most one of the two literals for each variable) and it is easy to see that the same holds for $v$ and $w$. Hence $S \cup A$ is a minimal dominating set of~$G$.

Given an instance $\varphi$ of SAT, we constructed in polynomial time an instance $(G,A)$ of $\Dcs{}$ that is equivalent to satisfiability of $\varphi$. This proves that $\Dcs{}$ is \NP{}-hard.
\end{proof}

\subsection{Limitations of the bicolored argument}\label{subsec:bic}

Let us present a brief argument of why enumerating the minimal dominating sets in a bicolored graph $G(A)$ is \DomEnum{}-hard if $A$ can contain an arbitrarily large clique and no restriction is put on the structure of $G-A$ nor its interactions with $A$. In other words, we argue that \DomEnum{} can be reduced to the problem of enumerating the minimal dominating sets in a bicolored graph $G(A)$ where $A$ is a clique.

Because of Theorem~\ref{thm:cobip-hard}, we know that enumerating minimal dominating sets of a co-bipartite graph $G$ is \DomEnum{}-hard. However, note that free to disregard the minimal dominating sets consisting of exactly one vertex in each clique of the partition, every minimal dominating set is included in one of the two cliques. 
Let $A_1$ and $A_2$ be the two sides of this partition.
Observe that as both $A_1$ and $A_2$ induce cliques, they satisfy any property that does not limit the size of the largest clique. Combined with the fact that minimal dominating sets consisting of exactly one vertex in each side of the partition are easy to enumerate, we obtain the desired conclusion.

Note however that this obstacle was circumvented in Theorem~\ref{th:pawdiam} by keeping track of what the forbidden structures in $G$ imply for the interactions between $G-A$ and $A$. Unfortunately, the arguments were quite ad hoc in nature and it is unclear how far they can be generalized.

This obstacle was bypassed in a different way in Theorem~\ref{thm:ktme}, simply by first enumerating all the minimal dominating sets without a given structure, then using the fact that the structure appears in any remaining dominating set to guess where it does, and finally arguing that the vertices that remain to be dominated cannot induce an arbitrarily large clique. We now show that this technique is in fact very limited.

\subsection{Limitations of enumerating all minimal dominating sets with a certain structure}

We present now a brief argument on why enumerating all $H$-free minimal dominating sets in a graph is \DomEnum{}-hard unless $H$ is a clique of size at most $2$. Here, a minimal dominating set $D$ is {\em{$H$-free}} if $G[D]$ does not contain $H$ as an induced subgraph.

The case when $H$ is not a clique is directly implied by the argument in Section~\ref{subsec:bic}. We now focus on the case when $H$ is a clique on at least $3$ vertices; it suffices to handle the case when $H$ is a triangle. In other words, we argue that \DomEnum{} can be reduced to the question of enumerating all triangle-free minimal dominating sets.

Consider a graph $G$. We build an auxiliary graph $G'$ by creating two copies $A$ and $B$ of $V(G)$, creating a vertex $u$, and setting $V(G')=A \cup B \cup \{u\}$. We set $A$ to be an independent set, $B$ to be a clique, and the vertex $u$ to be adjacent to all of $A$ and none of $B$. We set the edges between $A$ and $B$ as follows: a vertex in $A$ and a vertex in $B$ are adjacent if and only if the vertices of $G$ they originate from are the same or are adjacent.

Let us consider what the structure of a minimal dominating set $D$ of $G'$ can be, and how easy it is to generate all minimal dominating sets of a given type.
We consider three cases.
\begin{enumerate}
\item $u \not\in D$. We generate all minimal dominating sets of the split graph $G'[A \cup B]$: this can be done in output-polynomial time according to Proposition~\ref{prop:split-properties}. For each such minimal dominating set, either the intersection with $A$ is non-empty and it is a minimal dominating set of $G'$, or it is empty and we can generate in polynomial-time all additions of a vertex of $A$ that would result in a minimal dominating set of~$G'$, if any. Since the number of minimal dominating sets of $G'[A \cup B]$ with empty intersection with $A$ is polynomially bounded by the number of those with non-empty intersection (see Lemma~\ref{lem:triangle-free-candidates-enumeration}, Inequality~\eqref{eq:triangle-free-trashbound}), we can generate all minimal dominating sets of $G'$ not containing $u$ in output-polynomial time.
    \item \mbox{$u \in D$ and $D \cap B \neq \emptyset$}. Then $|D \cap B|=1$, and for any $v \in B$, the set $\{u,v\}$ is a minimal dominating set of $G'$.
    \item \mbox{$u \in D$ and $D \cap B = \emptyset$}. All these minimal dominating sets are triangle-free. We observe that there is a bijection between the minimal dominating sets of this type and the minimal dominating sets of $G$.
\end{enumerate}

Summarizing, the first two types of minimal dominating sets are easy to generate in output-polynomial time. We note that, free again to disregard minimal dominating sets that are easy to generate, enumerating all triangle-free minimal dominating sets of $G'$ boils down to enumerating all minimal dominating sets of $G'$ that are included in $A \cup \{u\}$ and contain $u$. This is equivalent to enumerating all minimal dominating sets of $G$, hence the conclusion.

Note, however, that there is still hope for this technique when we assume some structure on the whole graph.

\section{Perspectives for further research}\label{sec:concl}

In this paper, we investigated the enumeration of minimal dominating sets in graph classes forbidding an induced subgraph $H$.
We gave algorithms that run in output-polynomial time and polynomial space when $H$ is a clique, or more generally when $H =K_t + K_2$, and when $H$ is the paw or the diamond. We now discuss possible directions for future research. For simplicity, let us here denote by $\DomEnum{}(H)$ the problem \DomEnum{} restricted to $H$-free graphs.

The most natural continuation of our work is to seek output-polynomial time algorithms for $\DomEnum{}(H)$ for other choices of the graph~$H$. 
We discuss a possible classification of the graphs $H$
depending on whether $\DomEnum{}(H)$ admits an output-polynomial time algorithm, is \DomEnum{}-hard, or is not known to belong to one of these two cases.
We stress that the first two cases may not be disjoint as it is currently open whether \DomEnum{} admits an output-polynomial time algorithm in general. However, in the current state of the art, such a classification will highlight specific graph classes where the problem could be attacked more easily than in the general case.

Because of Theorem~\ref{thm:cobip-hard}, if $H$ is such that co-bipartite graphs form a subclass of $H$-free graphs then $\DomEnum{}(H)$ is \DomEnum{}-hard.
This includes the cases $H=C_t$ or $H=P_t$ with $t \geq 5$.
This is also true for any graph $H$ that has an independent set of size at least three, in particular all graphs $H$ that have at least three connected components and graphs with two connected components where one component has one non-edge.
Therefore, all the graphs $H$ with more than one connected component for which $\DomEnum{}(H)$ is not known to be \DomEnum{}-hard are of the form $H=K_p + K_q$ (where by $+$ we denote the disjoint union), for integers~$p,q \geq 1$. We gave an output-polynomial time algorithm for the case where $p\leq 2$ or $q\leq 2$ in Theorem~\ref{thm:ktme} and leave open the existence of such algorithms for $p,q \geq 3$.

Let us now focus on connected choices of~$H$.
Besides the case where $H$ is a clique, which we addressed with Theorem~\ref{thm:kt-free}, we settled the case where $H =K_t - e$ for $t=4$ (Theorem~\ref{thm:diamond-free}). For $t\in\{2,3\}$, $\DomEnum{}(H)$ is output-polynomial time solvable since $(K_t - e)$-free graphs then are, respectively, cliques and disjoint unions of cliques.
To the best of our knowledge, it is currently unknown whether
$\DomEnum{}(K_t - e)$ for $t \geq 5$ is \DomEnum{}-hard and whether it is output-polynomial time solvable.
We also considered graphs $H$ of the form $(K_t-\{uv,vw\})$ for $t\geq 3$, \ie, graphs obtained from a clique on $t$ vertices by removing two incident edges. 
When $t = 3$, $(K_t-\{uv,vw\})$-free graphs are exactly the complete multipartite graphs, for which an output-polynomial time algorithm can be obtained as in the proof of Lemma~\ref{lem:paw-free-candidates-enumeration}.
We dealt with the case $t=4$ in Theorem~\ref{thm:paw-free} and leave open the cases of larger~$t$.

Regarding the exclusion of specific graphs, we note that the status of $\DomEnum{}(P_t)$ is completely explored: either $t\leq 4$ and an output-polynomial time algorithm is known, or $t \geq 5$ and the problem is \DomEnum{}-hard, as noted above. Among graph classes defined by forbidding an induced cycle, we proved that
$\DomEnum{}(C_3)$ is output-polynomial time solvable by Theorem~\ref{thm:triangle-free} and noted above that $\DomEnum{}(C_t)$ is \DomEnum{}-hard for $t\geq 5$, so only $\DomEnum{}(C_4)$ remains to be classified. The graph $C_4$ is also the only graph on at most 4 vertices for which $\DomEnum{}(H)$ has not been classified yet.
Other graph classes that are closed by taking induced subgraphs and where no output-polynomial time algorithm for \DomEnum{} neither \DomEnum{}-hardness proof are known to include unit-disk graphs~\cite{kante2008encyclopedia, golovach2016chordalbip} and comparability graphs.

Another natural research direction is to optimize the running times of our algorithms or to prove that this is not possible.
Theorem~\ref{thm:ext} suggests that no improvement of our results can be obtained using backtrack search. 
We leave as an open problem whether there are polynomial delay algorithms for \DomEnum{} in the cases that we considered.

Finally, we note that the algorithm of Theorem~\ref{thm:kt-free} has been implemented in python/\allowbreak{}SageMath \cite{sagemath}.

\section*{Acknowledgements}

The authors wish to thank Paul Ouvrard for extensive discussions on the topic of this paper. We gratefully acknowledge support from Nicolas Bonichon and the Simon family for the organization of the $3^{\textrm{rd}}$ Pessac Graph Workshop, where part of this research was done. We also thank the organisers of the Dagstuhl Seminar 18421 on algorithmic enumeration~\cite{dagstuhl2019enumeration} where some ideas present in this paper have been discussed. Last but not least, we thank Peppie for her unwavering support during the work sessions.

\bibliographystyle{alpha}
\bibliography{main}

\end{document}